\documentclass[a4paper,10pt]{article}
\usepackage[utf8]{inputenc}

\usepackage{bbding}
\usepackage{amsmath,amssymb,amsthm}
\usepackage{hyperref}
\usepackage{graphicx}
\usepackage{verbatim}
\usepackage{geometry}
\usepackage{mathrsfs}
\usepackage{tikz}
\tikzstyle{every picture}=[level distance = 8mm, baseline=-0.5ex]
\tikzstyle{prop}=[shape=circle,minimum size=6mm, draw=black!80, fill=green!30]

\renewcommand{\d}{\text{d}}

\newcommand{\C}{\mathbb{C}}
\newcommand{\R}{\mathbb{R}}

\newcommand{\Z}{\mathbb{Z}}
\newcommand{\N}{\mathbb{N}}

\newcommand{\dsetminus}{//}
\newcommand{\resOm}{\widehat{\mathcal{R}}_\Omega}
\newcommand{\uniOm}{\widehat{\mathcal{U}}_\Omega}

\newcommand{\calK}{\mathcal{K}}

\newcommand{\scrS}{\mathscr{S}}

\newtheorem{thm}{Theorem}[section]
\newtheorem{claim}[thm]{Claim}
\newtheorem{lem}[thm]{Lemma}
\newtheorem{coro}[thm]{Corollary}

\newtheorem{prop}[thm]{Proposition}

\theoremstyle{definition}
\newtheorem{defn}[thm]{Definition}
\newtheorem{rk}[thm]{Remark}
\newtheorem{ex}[thm]{Example}

\newtheorem{quest}{Question}

\begin{document}

\title{Borel-\'Ecalle resummation of a two-point function}
\author{Pierre~J.~Clavier${}^{1,2}$\\
\normalsize \it $^1$  Institute of Mathematics, \\
\normalsize \it University of Potsdam,\\
\normalsize \it D-14476 Potsdam, Germany\\
~\\
\normalsize \it $^2$ Mathematics Laboratory, \\
\normalsize \it Technische Universit\"at, \\
\normalsize \it 10623 Berlin, Germany\\
~\\
\normalsize email: clavier@math.uni-potsdam.de}

\date{}

\maketitle

\begin{abstract} 
 We provide an overview of the tools and techniques of  resurgence theory used in the Borel-\'Ecalle resummation method, which we then apply to the 
 massless Wess-Zumino model.
 Starting from already known results on the anomalous dimension of the Wess-Zumino model, we solve its renormalisation group equation 
  for the two point 
 function in a space of formal series. We show that this solution is 1-Gevrey and that its Borel transform is resurgent.
 The Schwinger-Dyson equation of the model is then used to prove an asymptotic 
 exponential bound for the Borel transformed two point function on a star-shaped domain of a suitable ramified complex plane. This prove that the 
 two point function of the Wess-Zumino model is Borel-\'Ecalle summable.
\end{abstract}

\tableofcontents

\section{Introduction}

\subsection{State of the art and goals of the paper}

Recently, much  progress has been made towards a better analytic understanding of Quantum Field Theories (QFTs), both in perturbative and non perturbative approaches. To illustrate this, for the perturbative 
approach  let us quote  \cite{DaZh17}, where the authors proved that the regularised Feynman rules of a QFT have their image in a space of meromorphic families of distributions with linear poles. For non 
perturbative approaches, let us mention \cite{Pa19}, where an explicit solution to the Schwinger-Dyson equation is found,  a more powerful result than the ones previously obtained in \cite{BrKr99,Cl14} for 
similar Schwinger-Dyson equations. Of particular interest to us are the results of the articles \cite{BeMaVa19}, \cite{BeCl14} and \cite{BeCl16} which undertake a resurgent study of some Quantum Field 
Theories.

Resurgence theory was developed in the late 70s and early 80s almost single-handedly by Jean \'Ecalle \cite{Ecalle81,Ecalle81b,Ecalle81c}. It was initially applied to problems from the theory of dynamical systems 
(Dulac's problem, see \cite{Ec92}) but quickly found applications in other branches of mathematics. The theory of averages is  one important aspect of resurgence theory, developed by Fr\'ed\'eric Menous 
\cite{Me96,Menous} in the late 90s and much more recently by Emmanuel Vieillard-Baron \cite{VB14}.

Notwithstanding its successes, the theory of resurgence has not reached a very large audience until fairly recently. In the early 2010, In\`es Aniceto and Riccardo Schiappa started in \cite{AnSc13} a 
program  to apply various aspects of resurgence theory to physics. In particular they use alien calculus to compute non perturbative contributions to physical theories.

This approach has proved to be  {successful} and nowadays resurgence theory is becoming part of the standard toolkit available to physicists. Instead of listing the various topics and articles of the physics 
literature applying resurgence theory, we refer the reader to the review \cite{AnBaSc18} and  {to} the introductory article \cite{Do14} for a presentation of the physicists'  point of view {on} resurgence.

The theory of averages, which plays a central role in the aspects of resurgence theory used by physicists, is nevertheless somewhat hidden in their presentations of resurgence. {This theory is} a central tool {of} the 
Borel-\'Ecalle resummation process {and as such} is part of the  mathematical background necessary for the physicists' applications of resurgence. It has to be considered if one wishes to reach 
rigorous mathematical results {for} physical problems.

Bearing this in mind, the present paper has two main goals:
\begin{itemize}
 \item To present a self-contained introduction to the Borel-\'Ecalle resummation method, useful for the reader who wishes to apply it to specific problems,  e.g. those coming from physics.
 \item To illustrate the applicability of this method to physical problems, and in particular in QFT, in which context we explicitly work  out a resurgent analysis of a Wess-Zumino model.
\end{itemize}

\subsection{Summary of the paper}

The main result of the paper is the following result:
\begin{thm} \label{thm:main}
 The solution {$\tilde G(L,a)$} of the Schwinger-Dyson equation \eqref{eq:SDnlin} and the renormalisation group equation \eqref{eq:RGE} is Borel-\'Ecalle resummable along the positive real axis. For any real 
 value of the kinematic parameter $L$, the resummed function $a\mapsto G^{\rm res}(a,L)$ is analytic in the domain
 \begin{equation*}
  \left|a-\frac{1}{20L}\right| < \frac{1}{20L}.
 \end{equation*}
\end{thm}
Before  proving the theorem, we give a brief introduction to some aspects of Ecalle's resurgence theory, which allows to review crucial concepts of the Borel-\'Ecalle resummation method. We first define the 
Borel-Laplace resummation operator (Definition \ref{defn:Borel_Laplace}), which is generalised by the Borel-\'Ecalle resummation method. We then introduce the notion of resurgent functions (Definition 
\ref{defn:resum_function}) and state some bounds on their convolution products that will be of crucial use later on (Equation \eqref{eq:bound_star_shaped} and Theorem \ref{thm:bound_resu_Sauzin}). Finally we 
present the notion of well-behaved averages (Definition \ref{defn:well_behaved_averages}) which allows to state the main result of Borel-\'Ecalle resummation method, namely Theorem 
\ref{thm:Borel_Ecalle_resummation}.

In the next section we introduce the model we will focus on: the massless Wess-Zumino model. It is a massless supersymmetric model in four dimension. Supersymmetry prevents the need for a vertex 
renormalisation, thus drastically simplifying the Schwinger-Dyson equations. This makes this model a true QFT model simple enough to be used as a testing ground for the Borel-\'Ecalle resummation method. We 
first introduce the equations whose solution we will study using the tools of resurgence theory: a truncated version of the Schwinger-Dyson equation \eqref{eq:SDnlin} and the renormalisation group equation
\eqref{eq:RGE}. We then state the known results we base this study upon, in particular Theorem \ref{thm:resurgence_gamma}.

The following section focuses on the study of the renormalisation group equation. We start by building a solution to the renormalisation group equation in a space of formal series (Proposition 
\ref{prop:solution_RGE}). We then show that this solution is 1-Gevrey (Proposition \ref{prop:G_one_Gevrey}), and thus that its Borel transform is analytic in a disc around the origin. The main result of this 
section is Theorem \ref{thm:resurgence_two_points_function}, which states that the studied solution is indeed resurgent. Using basic results of the theory of resurgent function and previous results on the 
Wess-Zumino model, the proof of this statement is reduced to a proof of normal convergence of a series of functions. The proof of this property relies on Sauzin's non-linear analysis of resurgent functions 
\cite{Sa12}.

The last section is a study of the asymptotic behavior of the Borel transform of the solution of the Schwinger-Dyson equation and the renormalisation group equation. We are concerned with its behavior at 
infinity on {a star-shaped subset $\mathcal{U}_\Omega$ of $\C$ containing the origin that can be embedded into the ramified plane $\C\dsetminus\Omega$.} We first explain that a naive approach give the asymptotic bound 
\begin{equation*}
 |\hat G(\zeta,L)| \leq K\exp(c|\zeta|^2g(\zeta)L),
\end{equation*}
with $g$ an asymptotic bounds of the Borel transform of the anomalous dimension of the Wess-Zumino model. This bound is not satisfactory since numerical studies of \cite{BeCl14} suggest that $g$ does not 
vanishes at infinity. To obtain a better bound, and to study the asymptotic behavior of $g$, we need to make use of the Schwinger-Dyson equation. We start by expanding it (Equation 
\eqref{eq:SDE_Borel_expanded}) and finding bounds on the numbers that appear in the expansion (Lemma \ref{lem:bound_X_nm}). We then use the Schwinger-Dyson equation and the renormalisation group equation to find 
improved bounds on the functions whose series is the Borel transform of the two point function (Proposition \ref{prop:main_bound}).

In the last subsection of this paper, we prove an asymptotic bound for the Borel transform $\hat\gamma$ of the anomalous dimension of the Wess-Zumino model (Proposition \ref{prop:bound_gamma_infinity}). This 
can then be used together to derive an asymptotic bound for the two point function of the theory: Theorem \ref{thm:bound_two_point_infinity}. These bounds hold on {the star-shaped domain $\mathcal{U}_\Omega$ of} the ramified plane 
$\C\dsetminus\Omega$. Theorem \ref{thm:bound_two_point_infinity}, together with Theorem \ref{thm:resurgence_two_points_function}, implies that the solution of the renormalisation group equation and 
Schwinger-Dyson equation is Borel-\'Ecalle resummable (Corollary \ref{coro:BE_res}). The proof of Theorem \ref{thm:main} is {concluded} by Proposition \ref{prop:numerical_bound} which precises the analyticity 
domain of the resummed function.

\subsection{Some open questions}

This paper aiming at being a gentle introduction to the basic concepts of Borel-\'Ecalle resummation procedure we have tried to provide examples of the key concepts arising in the initial discussion. The rest 
of the paper, which is more technical, also has the objective to convince the reader that this procedure can actually be carried out in non trivial problems of mathematical physics. We therefore try to 
motivate each computation in order to guide the reader through the sometimes cumbersome computations. We feel that the results of this paper open new and exciting directions of  research which we now briefly 
describe.

The proof of the Borel-\'Ecalle summability of a simple, but non trivial QFT is a first step towards  physically more relevant models, a long-term goal being . 
\begin{quest}
 Are the Yang-Mills model Borel-\'Ecalle (accelero-)summable?
\end{quest}
There are still quite a few technical issues to be tackled before reaching this aim. For example, the Schwinger-Dyson equations of non supersymmetric models generally do not close. One then imposes one 
further equation to study the system, whose compatibility with the gauge symmetry is still open. Also, the study of a asymptotically free QFT would probably require the more sophisticated accelero-summation 
method {as suggested in \cite{BeCl18}}.

However, other less ambitious questions seem to be within short term reach. One concerns the transseries expansion of Borel-\'Ecalle resummed functions. There is a very precise analytical link between a Borel 
summable series and the associated Borel resummed function, known as Watson's theorem. {This classical result was generalised by Sokal (among others) in \cite{So80}. Sokal's generalisation of Watson's theorem, henceforth called Sokal-Watson theorem, is the one commonly used in QFT.} To the best of the author's knowledge, {the following question is still open}:
\begin{quest} \label{question:1}
 Is there a {Sokal-Watson's theorem} for the Borel-\'Ecalle resummation method?
\end{quest}
{The most general case of transseries of any level (which are, in principle, obtainable from the most general Borel-\'Ecalle resummation method) is still beyond the reach of resurgence theory. Let us precise this question on the example of a level one transseries, as the one found in the Wess-Zumino model studied below, and most of physical applications to resurgence:
\begin{equation*}
 \Phi(a) = \sum_{k=0}^{+\infty}\Phi_n(a)e^{-Sn/a}.
\end{equation*}
Provided the Borel transforms $\widehat\Phi_n$ all satisfy {simultaneously} the assumptions of the Sokal-Watson theorem, we readily obtain a positive answer. These assumptions are that each of the Borel transform $\widehat\Phi_n$ admits an analytic extension to a common strip containing the real line, and a common exponential bond at infinity on this strip. Then we would obtain common analytical properties and bounds for each of the $\Phi_n$ which can be used to prove the analyticity domain and bounds on $\Phi$.

However in general the Borel transforms do not admit an analytic extension in a strip containing the positive real line: they have singularities. They only admit analytic continuation along paths avoiding these singularities. This is the case we are dealing with for the Wess-Zumino model. It is more delicate but should, at least in principle, be solvable using Ecalle's resurgence relations which should allow us to reduce this case to the previous one.}

{A Sokal-Watson theorem for transseries} would be of importance for the physical implications of the Borel-\'Ecalle resummation method. Indeed, for these applications only a transseries expansions of the Borel-\'Ecalle resummed function 
were computed. These transseries are not the full Borel-\'Ecalle resummed functions but rather a good approximation which can then be compared to experimental results. A Watson's theorem for Borel-\'Ecalle 
resummation which would be formulated with transseries would provide a more precise meaning to the word ``good'' in the previous sentence and allow to have estimates for error margins coming from the 
truncations of the transseries.

Another reason why such a theorem would be of importance lies in the details of the physical applications of resurgence theory to physics. The coefficients of the transseries expansion are computed using the 
so-called median average, which can be expressed in terms of the alien derivatives of the formal series to be resummed. The median average is one special average, a notion that will be introduced below. 
However, it is not a ``well-behaved average'', which are the ones that should be used for the Borel-\'Ecalle resummation method. Nonetheless, one could expect the transseries expansion of a function to be 
unique. Thus Watson's theorem for Borel-\'Ecalle resummation would give a better mathematical ground to physical computations.

As we shall see later, in order to perform a Borel-\'Ecalle resummation on a formal series, a choice of a well-behaved average is required. This choice is not unique which raise a natural and important 
question:
\begin{quest} \label{question:2}
 How does the Borel-\'Ecalle resummed function depend on the choice of the well-behaved average? 
\end{quest}
{A priori, different well-behaved average being very different objects, one could assume the answer to be that different choices of average build different solutions. However, a finer analysis of the problem at hand leads to more subtle conclusions. In particular, for physics-related problems studied with resurgence, the Riemann surface on which the resurgent functions have their domain is highly structured. This should induce that different average coincide up to one (or possibly finitely many) free parameters, as in the case of ODEs \cite{Co98}, \cite{Co06}.}

{Therefore, one} could conjecture that it actually does not depends on the choice made and that changing averages amounts to a reparametrisation of the solution. This conjecture is motivated by an observation of 
\cite{Me97} that it indeed holds for a specific problem and from the fact that two averages are always related by a so-called passage automorphism. Even if the choice of the average changes the resummed 
function, one should expect stability of some physically relevant properties, for example the poles of the resummed function. This observation relates this question with question \ref{question:1}: one 
should not expect the transseries expansion to depend on a specific choice of an average.

One last important question lying outside the scope of the present article is 
\begin{quest}
 How and in which extent can one characterise the Borel-\'Ecalle resummed function solving a given problem?
\end{quest}
It was argued in \cite{BeCl14} that the Borel-\'Ecalle resummation method applied to QFT could give a non perturbative mass generation mechanism. In order to study the relevance of this mechanism, one needs 
to study the poles of the resummed function. {An attempt toward such a study is made at the very end of this paper.} This rather ambitious question is linked to the question \ref{question:1} and \ref{question:2}. Let us finally mention that this last question has motivated the 
present study.

\section{Elements of resurgence theory}

\subsection{The Borel-Laplace resummation method}

Many excellent introductions of the classical theory of Borel-Laplace resummation can be found in the literature. In particular, the PhD thesis 
\cite{Bo11} offers a well-written and short presentation of this topic (in French), while and the article \cite{Sa14} is a very thorough 
introduction. Nonetheless, {we shortly present the Borel-Laplace resummation method}  to obtain a self-contained paper.
\begin{defn}
The {\bf formal Borel transform} is defined on formal series as
 \begin{eqnarray*}
             \mathcal{B}: (z^{-1}~\mathbb{C}[[z^{-1}]],.) & \longrightarrow & (\mathbb{C}[[\xi]],\star) \\
  \tilde{f}(z) = \frac{1}{z}\sum_{n=0}^{+\infty}\frac{c_n}{z^n} & \longrightarrow & \hat{f}(\xi) = \sum_{n=0}^{+\infty}\frac{c_n}{n!}\xi^n.
 \end{eqnarray*}
\end{defn}
The formal Borel transform enjoys many useful properties, easy to prove by manipulation of formal series (see for example \cite{Sa14}, 
\S 4.3 and 5.1).
\begin{prop} 
 Let $\tilde{f}(z),\tilde{g}(z)\in z^{-1}~\mathbb{C}[[z^{-1}]]$ be two formal series and $\hat{f},\hat{g}\in\mathbb{C}[[\xi]]$ be their Borel 
 transforms. Then the following hold
 \begin{itemize}
  \item $\mathcal{B}(\tilde f.\tilde g) = \hat f\star \hat g$;
  \item $\mathcal{B}(\partial\tilde f) = -\zeta\hat f$;
  \item $\mathcal{B}(z^{-1}\tilde f) = \int\hat f$;
  \item if $\tilde{f}(z)\in z^{-2}~\mathbb{C}[[z^{-1}]]$, then $\mathcal{B}(z\tilde f)=\frac{d\hat f}{d\zeta}$;
 \end{itemize}
 where the derivatives and the integral are formal (i.e. defined term by terms) and $\star$ stands for the convolution product of formal 
 series. These properties stay true in the case where $\hat f, \hat g$ are convergent. In this case, the first property becomes
 \begin{equation*}
  \mathcal{B}(\tilde f.\tilde g)(\zeta) = (\hat f\star \hat g)(\zeta) = \int_0^\zeta\hat f(\eta) \hat g(\zeta-\eta)d\eta
 \end{equation*}
 for $\zeta$ in the intersection of the convergence domains of $\hat f$ and $\hat g$.
\end{prop}
We will in fact study the case where the Borel transform is convergent. There exists a simple necessary and sufficient 
condition of the convergence of the Borel transform, but we need one more definition.
\begin{defn}
 A formal series $\tilde f(z) = \frac{1}{z}\sum_{n=0}^{+\infty}\frac{a_n}{z^n}$ is {\bf 1-Gevrey} if
  \begin{equation*}
   \exists A,B>0:~|a_n|\leq AB^n n!~~\forall n\in\N.
  \end{equation*}
  In this case, we write $\tilde f(z)\in z^{-1}\C[[z^{-1}]]_1$.
\end{defn}
An easy but important result (see for example \cite{Sa14}, \S 4.2) is then
\begin{thm}
 Let $\tilde{f}(z)\in z^{-1}~\mathbb{C}[[z^{-1}]]$ be a formal series. Its Borel transform has a {strictly positive, possibly infinite} radius of convergence (in this case 
 we write $\hat f\in\C\{\zeta\}$) if and only if $\tilde f$ is 1-Gevrey.
\end{thm}
One can make other statements relating for example the Borel transform $\hat f$ and the convergence of its associated formal series 
$\tilde f$, however such considerations will play no role here. The importance of the Borel transform for us lies in particular in the 
existence of an inverse operation: the Laplace transform.
\begin{defn}
 Let $\theta\in[0,2\pi[$ and set $\Gamma_\theta:=\{Re^{i\theta},R\in[0,+\infty[\}$. Let $\hat f\in\C\{\zeta\}$ be a germ admitting an 
 analytic continuation in an open subset of $\C$ containing $\Gamma_\theta$ and such that 
 \begin{equation} \label{eq:exp_bound}
  \exists c\in\R,~K>0:|\hat f(\zeta)|\leq K e^{c|\zeta|}
 \end{equation}
 for any $\zeta$ in $\Gamma_\theta$. Then the {\bf Laplace transform } of $\hat f$ in the direction $\theta$ is defined as
 \begin{equation*}
  \mathcal{L}^{\theta}[\hat f](z) = \int_0^{e^{i\theta}\infty}\hat f(\zeta)e^{-\zeta z}\d\zeta.
 \end{equation*}
\end{defn}
The Laplace integral of this definition is finite for $z$ in an open subset of $\C$ to be specified later on. For the time being, let us say that
the composition of the Laplace and the Borel transforms is the so-called Borel-Laplace resummation method.
\begin{defn} \label{defn:Borel_Laplace}
 Let $\theta\in[0,2\pi[$, and $\tilde f(z)\in z^{-1}\C[[z^{-1}]]_1$ such that the Laplace transform of its Borel transform exists in the 
 direction $\theta$. Then $\tilde f(z)$ is said to be {\bf Borel summable} in the direction $\theta$. 
 
 The {\bf Borel-Laplace resummation operator} in the direction $\theta$ is defined on the functions Borel summable in the direction $\theta$ 
 as 
 \begin{equation*}
  S_{\theta} = \mathcal{L}^{\theta}\circ\mathcal{B}.
 \end{equation*}
 For a Borel summable formal series $\tilde f$, the function $z\mapsto S_\theta[f](z)$ is called the {\bf Borel sum} of $\tilde f$.
\end{defn}
Varying the direction $\theta$ of the resummation leads to interesting concepts and phenomena such as sectorial resummation and the 
Stokes phenomenon, however we will not {deal with} them here. 
\begin{rk}
 It is easy to see that a formal series $\tilde f(z)\in z^{-1}\C[[z^{-1}]]$ with a finite non-zero radius of convergence has a Borel transform admitting
 an exponential 
 bound \eqref{eq:exp_bound} for all $\theta\in[0,2\pi[$ and that its the Borel sum in any direction coincide with the usual summation of 
 series. Thus the Borel-Laplace resummation method is an extension of the usual summation of series.
\end{rk}
We claimed in Definition \ref{defn:Borel_Laplace} that the Borel sum of a Borel summable function is a function. This is a consequence of 
the following theorem, which is itself a consequence of classical results of the theory of the Laplace transformation.
\begin{thm}
 Let $\tilde f$ be a formal series, Borel summable in the direction $\theta$ with the exponential bound \eqref{eq:exp_bound}:
 \begin{equation*}
  \exists c\in\R,~{K>0}:|\hat f(\zeta)|\leq K e^{c|\zeta|}.
 \end{equation*}
 Then its Borel sum is analytic as a function of $z$ in the half-plane $\Re(ze^{i\theta})>c$.
\end{thm}
We have seen that one can perform the Borel-Laplace resummation method in non-singular directions of the Borel transform only. However, in 
many problem of interest (in particular, of interest to physicists), the Borel transform will have singularities in the direction where we wish to perform 
the resummation. \'Ecalle defines objects where the poles have a specified location (resurgent functions) and objects allowing to 
compute these singularities ({alien} derivatives). The introduction of these concepts is the subject of the remaining part of this section.

\subsection{Resurgent functions} \label{subsec:res_fct}

In the rest of this text, we take $\Omega$ a non-empty, discrete and closed subset of $\C$. We recall that a function (or a germ) $f$ holomorphic
in a disc $D$, around the origin is continuable along a path $\gamma$ in $\C$ starting within the disc of convergence of the function if there is a finite 
family $(D_i)_{i\in \{1,\cdots,n\}}$ of convex open subset of $\C$ covering $\gamma$ such that $f$ is analytically continuable to $D\cup D_1\cup\cdots\cup D_n$.
\begin{rk}
 Being continuable along a path is much less strict than being continuable. In particular, a function continuable along a family of paths can be seen 
 as a function over an open subset of a cover of $\C$ rather than a function of $\C$. 
\end{rk}
\begin{defn} \label{defn:resum_function}
 A germ $\phi\in\C\{\zeta\}$ is said to be an {\bf $\Omega$-resurgent function} if it is continuable along 
  any rectifiable (i.e. of finite length) path in $\C\setminus\Omega$. We set
  \begin{equation*}
   \widehat{\mathcal{R}}_\Omega := \{\text{all $\Omega$-continuable germs}\} \subset\C\{\zeta\}.
  \end{equation*}
\end{defn}
Now, the convolution product of two $\Omega$-resurgent function is well-defined inside the intersection of their convergence discs. A difficult question 
is whether or not this convolution product defines an $\Omega$-resurgent function. 
The following theorem is a cornerstone of resurgence theory, as it states when this is indeed the case 
and thus {when} resurgent functions are stable under an extension of the convolution 
product {and are therefore suited} to the study of non-linear differential equations.
\begin{thm}({Ecalle}, Sauzin \cite{Sa14}[Theorem 21.1]) \label{thm:stability_resu_fct}\\
 Let $\Omega\subset\C$ be non-empty, discrete and closed. Then $\widehat{\mathcal{R}}_\Omega$ is stable under the convolution product if, and only if, 
 $\Omega$ is closed under addition.
\end{thm}
The {classical} example below is already enough to show that $\Omega$ being closed under addition is a necessary condition. The hard part of the Theorem is thus to 
show that it is sufficient.
\begin{ex} \label{ex:comvo_resu}
Take $\omega_1,\omega_2\in\Omega$ {non zero and not proportional} and two meromorphic (and therefore resurgent) functions defined by
 $\hat{f}_1(\zeta) = \frac{1}{\zeta-\omega_1}$, $\hat{f}_2(\zeta) = \frac{1}{\zeta-\omega_2}$.
 Then a direct computation gives
 \begin{align*}
  (\hat{f}_1\star\hat{f}_2)(\zeta):= & \int_0^\zeta \hat{f}_1(\eta)\hat{f}_2(\zeta-\eta)d\eta \\
                                  = & \frac{1}{\zeta-\omega_1-\omega_2}\left[\int_0^\zeta\frac{d\eta}{\eta-\omega_1} + \int_0^\zeta\frac{d\eta}{\eta-\omega_2}\right] \\
 \end{align*}
 {Now, take $R$ the Riemann surface obtained from removing the lines $\omega_i[1,+\infty[$ from $\C$.}
 One can {then} check that the R.H.S. has indeed a pole in $\omega_1+\omega_2$ {on some sheets of the Riemann surface $R$}. Therefore if $\omega_1+\omega_2$ is not an element of $\Omega$, $\hat{f}_1\star\hat{f}_2$
 is not $\Omega$-resurgent. {For more details, see for example \cite[Section 20]{Sa14}.}
\end{ex}
 For $\Omega\subset\C$ non-empty, discrete and closed we set $\rho(\Omega):=\min\{|\omega|:\omega\in\Omega^*\}$, with $\Omega^*=\Omega$ if 
 $0\notin\Omega$ (this will be the case we will work with) and $\Omega^*=\Omega\setminus\{0\}$ otherwise.

Finally, we will use here bounds on convolution products of resurgent functions. First, recall that for any open set $U\subset\C$ containing the origin and 
star-shaped with respect to the origin, the following bound holds by direct computation:
\begin{equation} \label{eq:bound_star_shaped}
 |(\hat\phi_1\star\cdots\star\hat\phi_n)(\zeta)|\leq \frac{|\zeta|^{n-1}}{(n-1)!}\max_{[0,\zeta]}|\hat\phi_1|\cdots\max_{[0,\zeta]}|\hat\phi_n|
\end{equation}
for any $\hat\phi_1,\cdots,\hat\phi_n$ holomorphic on $U$ and $\zeta$ in $U$. We used $[0,\zeta]$ to denote the straight line between $0$ and $\zeta$.

This bound will be useful to show that the two-point function has the right type of bound at infinity on {$\mathcal{U}_\Omega$, the connected star-shaped domain in $\C$ obtained from placing a radial cut starting from the first singularities of $\Omega$,} and converges near the 
origin. However, it will not allow {us} to prove that it is resurgent. For this we will need to prove the normal convergence of a series of analytic 
continuations of functions along any paths avoiding $\Omega$. It will require the refined results of \cite{Sa12}, specific to resurgence 
theory. In order to state {the main} result, we need to introduce some notations, the same as in \cite{Sa12}.

First, let $\mathscr{S}_\Omega$ be the set of homotopy classes with fixed endpoints of path $\gamma:[0,l]\longrightarrow\C\setminus\Omega^*$ 
such that $\gamma(0)=0$. Then, for $\delta,L\geq0$ we set 
\begin{equation*}
 \calK_{\delta,L}(\Omega) := \{\zeta\in\scrS_\Omega|\exists\gamma\in\scrS_\Omega:\gamma(l)=\zeta,~\gamma\text{ of length }\leq L, \text{ dist}(\gamma(t),\Omega^*)\geq\delta~\forall t\in[0,l]\}.
\end{equation*}
It was shown in \cite{KaSa16} that $\mathscr{S}_\Omega$ has the structure of a Riemann surface, which is a cover of $\C\setminus\Omega$. 
Then $\calK_{\delta,L}(\Omega)$ can be described as the set of point of this cover which can be reached by paths of length less than 
$L$\footnote{in order to avoid confusion between the kinematic parameter of the two-points function and the length of the path we will 
denote the former by the letter $\Lambda$.} and staying at a distance at least $\delta$ of $\Omega^*$. One 
can in particular see the set of $\Omega$-resurgent functions as the set of locally integrable maps $f:\mathscr{S}_\Omega\longrightarrow\C$. This 
observation will become important to define the notion of average.
\begin{thm} \label{thm:bound_resu_Sauzin} \cite[Theorem 1]{Sa12}\\
 Let $\Omega\subset\C$ be discrete, closed and stable 
 under addition. Let $\delta,L>0$ with 
 $\delta<\rho(\Omega)$. 
 Set
 \begin{equation*}
  C:=\rho(\Omega)\exp\left(3+\frac{6L)}{\delta}\right), \qquad \delta':=\frac{1}{2}\rho(\Omega)\exp\left(-2-\frac{4L}{\delta}\right), \qquad L':=L+\frac{\delta}{2};
 \end{equation*}
  Then, for any any $n\geq1$ and $\hat\phi_1,\cdots,\hat\phi_n\in\widehat{\mathcal{R}}_\Omega$
 \begin{equation} \label{eq:bound_conv_resu}
  \max_{\calK_{\delta,L}(\Omega)}|\hat\phi_1\star\cdots\star\hat\phi_n| \leq \frac{2}{\delta}\frac{C^n}{n!} \max_{\calK_{\delta',L'}(\Omega)}|\hat\phi_1|\cdots\max_{\calK_{\delta',L'}(\Omega)}|\hat\phi_n|.
 \end{equation}
\end{thm}
\begin{rk}
 In subsequent work \cite{KaSa16}, Sauzin and Kamimoto have generalised this result to the cases where $\Omega$ is not stable under addition. One could 
 in principle use the result of \cite{KaSa16} to prove resurgence of the two-point functions on $\Z^*/3$ rather than $\N^*/3$, 
 However this is not needed for the Ecalle-Borel resummation procedure along the positive real axis, and we will satisfy ourselves with {exploiting} the above 
 bound, which is simpler to use.
\end{rk}
Theorem \ref{thm:bound_resu_Sauzin} implies that the convolution product is bicontinuous for the natural topology induced by the family of semi-norms
\begin{equation*}
 ||\hat\phi||_{\delta,L} := \max_{\zeta\in\calK_{\delta,L}(\Omega)}|\hat\phi(\zeta)|.
\end{equation*}
More precisely we have
\begin{coro} \cite[Theorem 2, Remark 3.2]{Sa12} \\
 $(\widehat{\mathcal{R}}_\Omega,\star)$ is a Fr\'echet algebra.
\end{coro}

\subsection{Borel-\'Ecalle resummation method}

In practice we do not need to consider path going backward to perform a Borel-\'Ecalle resummation. To simplify the statements we take from now on $\Omega$ to be 
a subset of $\R_+^*$.
\begin{defn}
 Let $\C\dsetminus\Omega$ be the {\bf$\Omega$-ramified plane}, namely the space of homotopy classes $[\gamma]$ of rectifiable paths 
 $\gamma:[0,1]\mapsto\C\setminus\Omega$ such that $\forall t,t'\in[0,1],~t<t'\Rightarrow \Re(\gamma(t))< \Re(\gamma(t'))$.
\end{defn}
One can show that $\C\dsetminus\Omega$ has the structure of a Riemann surface, see \cite{KaSa16}. $\C\dsetminus\Omega$ is a cover of 
$\C\setminus\Omega$. We call $\pi:\C\dsetminus\Omega\longrightarrow\C\setminus\Omega$ the 
canonical local biholomorphism associated to this Riemann surface. We refer the reader to \cite[Section 3]{KaSa16} for a precise definition of this 
geometric object. We omit these definitions as they will play only a minor role in the present work.

Let $\zeta\in\C\setminus\Omega$ and $\underline\zeta\in\C\dsetminus\Omega$ such that $\pi(\underline{\zeta})=\zeta$. 
If $\Omega=\{\omega_1,\omega_2,\cdots\}\subset\R_+^*$ with $\omega_0:=0<\omega_1<\omega_2<\cdots$, we write $\zeta^{\epsilon_1,\cdots,\epsilon_n}$ 
instead of $\underline\zeta$, with 
$\epsilon_i\in\{+,-\}$, $(\epsilon_1,\cdots,\epsilon_n)$ the signature of the branch of $\C\dsetminus\Omega$ on which 
$\zeta^{\epsilon_1,\cdots,\epsilon_n}$ stands and 
$|\pi(\zeta^{\epsilon_1,\cdots,\epsilon_n})|\in]\omega_n,\omega_{n+1}[$. 

{We further write $\mathcal{U}_\Omega$ the connected star-shaped domain in $\C$ obtained from placing a radial cut starting from $\omega_1$, the first singularity of $\Omega$. We identify $\mathcal{U}_\Omega$ with a subset of $\C\dsetminus\Omega$: for $\zeta\in\mathcal{U}_\Omega$, if $\Re(\zeta)\leq\omega_1$, we identify it with the homotopy class of the straight path from the origin to $\zeta$. Otherwise we identify it with $\zeta^{+,\cdots,+}$ if $\Im(\zeta)>0$ and with $\zeta\leftrightarrow\zeta^{-,\cdots,}$ if $\Im(\zeta)<0$. Notice that in the more common case where $\Omega\subset\R^*$, $\mathcal{U}_\Omega$ is obtained from placing two radial cuts at the first singularities of $\Omega$.
}

From now one we will make the simplifying assumption that $\Omega=\omega\N^*$ 
for some $\omega\in\R_+^*$. 

While performing a Borel-\'Ecalle resummation, it will be useful to see $\Omega$-resurgent functions as locally integrable functions from 
$\C\dsetminus\Omega$ to $\C$.
We also denotes by 
$\widehat{\mathcal{U}}_\Omega$ the set of uniform functions on $\C\dsetminus\Omega$; i.e. the set of functions $\hat\phi$ whose value at $\zeta$ do not depend on 
the branch of $\C\dsetminus\Omega$ $\zeta$ sits on:
\begin{equation*}
 \forall (\zeta,\zeta')\in(\C\dsetminus\Omega)^2,~\pi(\zeta)=\pi(\zeta')\Longrightarrow\hat\phi(\zeta)=\hat\phi(\zeta').
\end{equation*}
\begin{defn}
 An {\bf average} is a linear map ${\bf m}:\resOm\longrightarrow\uniOm$ defined by its weights $\{{\bf m}^{\varepsilon_1,\cdots,\varepsilon_n}\}$ and its 
 action on resurgent functions: for any such $\phi$ and any $\zeta\in\C\setminus\Omega$ with $|\zeta|\in[n\omega,(n+1)\omega[$
  \begin{equation*}
   ({\bf m}\phi)(\zeta) = \sum_{\varepsilon_1,\cdots,\varepsilon_n=\pm}{\bf m}^{\varepsilon_1,\cdots,\varepsilon_n}\phi(\zeta^{\varepsilon_1,\cdots,\varepsilon_n});
  \end{equation*}
  with the coherence relations ${\bf m}^\emptyset=1$ and 
  \begin{align*}
   {\bf m}^{\varepsilon_1,\cdots,\varepsilon_{n-1},+} & + {\bf m}^{\varepsilon_1,\cdots,\varepsilon_{n-1},-} = {\bf m}^{\varepsilon_1,\cdots,\varepsilon_{n-1}}\\
   {\sum_{\varepsilon_i=\pm}{\bf m}^{\varepsilon_1,\cdots,\varepsilon_i,\cdots,\varepsilon_{n}}} & {= {\bf m}^{\varepsilon_1,\cdots,\varepsilon_{i-1},\varepsilon_{i+1},\cdots,\varepsilon_{n}} \quad\forall i\in\{1,\cdots,n-1\}}
  \end{align*}
  {with ${\bf m}^{\varepsilon_1,\cdots,\varepsilon_{i-1},\varepsilon_{i+1},\cdots,\varepsilon_{n}}$ in the last condition an average over $\Omega\setminus\{\omega_i\}$ if $\Omega=\{\omega_j|j\in\N^*\}$ with the convention $0<\omega_1<\cdots<\omega_i<\cdots$ as before.}
\end{defn}
It is a simple exercise to check that the following are examples of averages.
\begin{ex} \label{ex:averages}
  \begin{itemize}
  \item Left lateral average:
  \begin{align*}
   {\bf mul}^{\varepsilon_1\cdots\varepsilon_n} = \begin{cases}
                                                   & 1 \quad \text{ if }\varepsilon_1=\cdots=\varepsilon_n=+ \\
                                                   & 0 \quad \text{ otherwise.}
                                                  \end{cases}
  \end{align*}
  \item Median average:
  \begin{equation*}
   {\bf mun}^{\varepsilon_1\cdots\varepsilon_n} = \frac{(2p)!(2q)!}{4^{p+q}(p+q)!p!q!}
  \end{equation*}
  with $p$ (resp. $q$) the number of $+$ (resp. of $-$) in $\{\varepsilon_1,\cdots,\varepsilon_n\}$.
  \item  Catalan average: Let 
  $Ca_n$ be the $n$-th Catalan number, $Qa_n(x)$ the $n$-th Catalan polynomial, $\alpha,\beta\in\R$, $\alpha+\beta=1$. 
  
  Write $\pmb{\varepsilon}=\varepsilon_1\cdots\varepsilon_n =(\pm)^{n_1}(\mp)^{n_2}\cdots(\varepsilon_s)^{n_s}$, set 
  \begin{equation*}
   {\bf man}_{(\alpha,\beta)}^{\pmb \varepsilon} = (\alpha\beta)^nCa_{n_1}\cdots Ca_{n_{s-1}}Qa_{n_s}((\alpha/\beta)^{\varepsilon_n}).
  \end{equation*}
 \end{itemize}
\end{ex}
The notion of average is too weak to be used as such. Indeed, we want the averaged function {\bf m}$\phi$ to 
\begin{itemize}
 \item Solve the same equation as $\phi$;
 \item Be a real function\footnote{$f:U\subset\C\longrightarrow\C$ is real if $f(\bar z) =\overline{f(z)}$ whenever both sides of the equation make sense. 
 We require this condition since we want the resummed function to represent a physical quantity.}
 \item To admit a Laplace transform provided that $\phi$ had a reasonable behavior at infinity.
\end{itemize}
These requirements are formalized by the notion of well-behaved average.
\begin{defn} \label{defn:well_behaved_averages}
 An average {\bf m} is called {\bf well-behaved} if
  \begin{itemize}
   \item {\bf (P1)} It preserves the convolution ${\bf m}(\phi\star\psi) = ({\bf m}\phi)\star({\bf m}\psi)$.
   \item {\bf (P2)} It preserves reality: ${\bf m}^{\varepsilon_1\cdots\varepsilon_n} = \overline{\bf m}^{\bar\varepsilon_1\cdots\bar\varepsilon_n}$, 
   with $\bar\pm=\mp$.
   \item {\bf (P3)} It preserves exponential growths: 
   $\forall\phi\in\resOm,\zeta\in\C\setminus\Omega$
   \begin{equation*}
   |\phi(\zeta^{\pm\cdots\pm})|\leq Ke^{c|\zeta|} ~ \Longrightarrow ~ |({\bf m}\phi)(\zeta)|\leq  Ke^{c|\zeta|}
   \end{equation*}
  \end{itemize}
\end{defn}
{\begin{rk}
     The $\zeta^{\pm\cdots\pm}$ appearing in condition {\bf (P3)} can be seen as an element of $\mathcal{U}_\Omega$. So condition {\bf (P3)} can be formulated as: the average preserves exponential growth along $\mathcal{U}_\Omega$.
    \end{rk}
}
\begin{rk}
 In general the equation one is studying with the Borel-\'Ecalle resummation method is a differential equation. However, averages naturally preserve the 
 differential structure: since $\mathcal{B}(\partial_z\tilde f)(\zeta)=-\zeta\hat f(\zeta)$ and since $\zeta\mapsto-\zeta$ is in $\uniOm$, 
 $\mathcal{L}[{\bf m}\mathcal{B}(\partial_z\tilde f)](z)=\partial_z\mathcal{L}[{\bf m}\mathcal{B}(\tilde f)](z)$. We used the variable $z=1/a$ for the 
 Borel transform for simplicity.
\end{rk}
The following table lists the properties of the averages of Example \ref{ex:averages}.
\begin{center}
\begin{tabular}{|l||c|c|c|}
\hline
             & {\bf (P1)} & {\bf (P2)} & {\bf (P3)} \\ \hline
 ${\bf mul}$ & {\Checkmark} & N & \Checkmark  \\ \hline
 ${\bf mun}$ & \Checkmark & \Checkmark & N  \\ \hline 
 ${\bf man}$ & \Checkmark & \Checkmark & \Checkmark  \\ \hline
\end{tabular}
\end{center}
In particular, the fact that the Catalan average is a well-behaved average is a highly non-trivial result of \cite{Menous}. A finite number of other 
families of well-behaved averages are known. It is conjectured there are no more than the ones already known. Progresses toward a classification of 
well-behaved averages were recently made in \cite{VB19}, using methods from the theory of Rota-Baxter algebras.
{\begin{rk}
     One can use an average that has only conditions {\bf (P1)} and {\bf (P2)} when one works on a problem with only one (or eventually finitely many) alien derivatives acting non-trivially on the Borel transform. This is the case in non-trivial problems, e.g. for non-linear systems of ODEs of rank one studied in \cite{Co06} or the Schwinger-Dyson equation of the Yukawa model studied in \cite{BoDu20}.
     
     As soon as the an infinite number of alien derivatives act non-trivially on the Borel transform, condition {\bf (P3)} is needed, as implied by Lemma 5 and Equation (40) of \cite{EcMe95}. We will see in the next Section that we are in this case.
    \end{rk}
}

Finally, the core of the Borel-\'Ecalle resummation method can be summed up in the following theorem:
\begin{thm} \label{thm:Borel_Ecalle_resummation}
 Let $(E)$ a differential equation admitting a solution $\tilde f\in\C[[a]]_1$ such that $\hat f\in\resOm$ for some $\Omega=\omega\N^*\subset\R^*_+$ 
 and $|\phi(\zeta^{\pm\cdots\pm})|\leq Ke^{c|\zeta|}$ for $|\zeta|$ big enough. Let {\bf m} be a well-behaved average. Then 
 \begin{equation*}
  f^{\rm res}:=\mathcal{L}\circ{\bf m}\circ\mathcal{B}\circ\tilde f
 \end{equation*}
 is a solution of $(E)$ analytic in the open set
 \begin{equation*}
  U =\{a\in \C: |a-c/2|<c/2\}.
 \end{equation*} 
\end{thm}

\section{The Wess-Zumino model}

We introduce the model we are going to study and state some known facts about it. Some of these results are well-known (e.g. the derivation of 
the Schwinger-Dyson and renormalisation group equations) while other are more recent. These can all be found in the PhD thesis \cite{Cl15}.

\subsection{Presentation of the model}

The Wess-Zumino is one of the simplest possible supersymmetric model: it is massless and exactly supersymmetric. It was first introduced and 
studied in the papers \cite{WeZu74a,WeZu74b}, seminal to supersymmetry. This model has two features that make it suitable as a first QFT to study 
within the framework of resurgence theory. 

First, the $\beta$ and $\gamma$ functions are proportional: $\beta=3\gamma$. This can in particular be shown using Hopf-algebraic techniques. It also presents the 
striking feature that it needs no vertex renormalisation, due to its (exact) supersymmetry. Therefore the Schwinger-Dyson equation for the 
two point function, truncated to the first loop, actually decouples from the Schwinger-Dyson equations for higher point functions. It reads
\begin{equation}\label{eq:SDnlin}
\left(
\tikz \node[prop]{} child[grow=east] child[grow=west];
\right)^{-1} = 1 - a \;\;
\begin{tikzpicture}[level distance = 5mm, node distance= 10mm,baseline=(x.base)]
 \node (upnode) [style=prop]{};
 \node (downnode) [below of=upnode,style=prop]{}; 
 \draw (upnode) to[out=180,in=180]   
 	node[name=x,coordinate,midway] {} (downnode);
\draw	(x)	child[grow=west] ;
\draw (upnode) to[out=0,in=0] 
 	node[name=y,coordinate,midway] {} (downnode) ;
\draw	(y) child[grow=east]  ;
\end{tikzpicture}.
\end{equation}
The other equation we are going to study is the renormalisation group equation. It takes the particularly simple form
\begin{equation} \label{eq:RGE}
 \partial_LG(L,a) = \gamma(a)(1+3a\partial_a)G(L,a)
\end{equation}
with $\gamma(a):=\partial_L {G(L,a)}|_{L=0}$ the anomalous dimension of the theory.

Using some known results of this model which are going to be listed in the next subsection, we will study the system composed of the renormalisation 
group and the Schwinger-Dyson equations. Let us emphasize that this study will be purely mathematical. Within the assumption that this study actually 
carries most of the information of the non perturbative regime of the Wess-Zumino model, we will then derive some physical interpretations of our work at 
the very end of this paper. 

{As already stated, we will in this paper rigorously prove the Borel-\'Ecalle summability of the solution of the system \eqref{eq:SDnlin}-\eqref{eq:RGE}. For this, we will start from results of previous articles \cite{BeCl16,BeCl13,BeCl14} and use in particular analytical tools for resurgent functions {developed} in \cite{Sa12,Sa14}. To the best of the author's knowledge, it is the first time these tools are applied to {an} analyse a QFT model. However, other methods also allow to analyse summability of PDEs, see for example \cite{Co08,CoTa07} and references therein. These other methods could also be used in future studies of other QFT models.

}

\subsection{State of the art}

Writing $G$ as a formal series in $L$
\begin{equation} \label{eq:L_expansion_G}
 G(L,a) = \sum_{k=0}^{+\infty}\gamma_k(a)\frac{L^k}{k!},
\end{equation}
(with $\gamma_0(a)=1$ and $\gamma_1(a)=:\gamma(a)$)
we can easily write the RGE \eqref{eq:RGE} as an induction relation on the $\gamma_k$s:
\begin{equation} \label{eq:recursion_gamma}
 \gamma_{k+1}(a) = {\gamma(a)}(1+ 3a\partial_a)\gamma_k.
\end{equation}
This justifies that we look for an equation {for} $\gamma$ rather than an equation {for} $G$. Plugging the expansion \eqref{eq:L_expansion_G} into the 
Schwinger-Dyson equation \eqref{eq:SDnlin} and computing the Feynman integral we obtain
\begin{equation} \label{SDE}
 \gamma(a) = \left.a\left(1+\sum_{n=1}^\infty\frac{\gamma_n(a)}{n!}\frac{d^n}{dx^n}\right)\left(1+\sum_{m=1}^\infty\frac{\gamma_m(a)}{m!}\frac{d^m}{dx^m}\right)H(x,y)\right|_{x=y=0},
\end{equation}
with $H$ the one-loop Mellin transform:
\begin{align}
 H(x,y) & = \frac{\Gamma(1+x)\Gamma(1+y)\Gamma(1-x-y)}{\Gamma(1-x)\Gamma(1-y)\Gamma(2+x+y)} \label{eq:Mellin0} \\
	& = \frac{1}{1+x+y}\exp\Bigl(2\sum_{k=1}^{+\infty}\frac{\zeta(2k+1)}{2k+1}\left((x+y)^{2k+1}-x^{2k+1}-y^{2k+1}\right)\Bigr). \label{eq:Mellin}
\end{align}
We will study the Borel transform of {Equation \ref{SDE}}. It maps the usual product of formal series to a convolution product and the identity function to the 
constant function $\zeta\mapsto1$. Separating the $1$ in the equation above from the rest we end up with 
\begin{equation*}
 \hat\gamma(\zeta) = 1 + \left.2\sum_{n=1}^\infty \frac{(1\star\hat\gamma_n)(\zeta)}{n!}\frac{d^n}{dx^n}H(x,y)\right|_{x=y=0} + \left.\sum_{n,m=1}^\infty \frac{(1\star\hat\gamma_n\star\hat\gamma_m)(\zeta)}{n!m!}\frac{d^n}{dx^n}\frac{d^m}{dy^m}H(x,y)\right|_{x=y=0}.
\end{equation*}
Similarly, taking the Borel transform of the renormalisation group equation \eqref{eq:recursion_gamma} one obtains
\begin{equation} \label{eq:RGE_borel}
 \hat\gamma_{n+1}(\zeta) = \hat\gamma\star(4+3\zeta\partial_\zeta)\hat\gamma_n(\zeta).
\end{equation}
Now, $\gamma(a)$ is a formal series with coefficients in $\C$, without constant term:
\begin{equation*}
 \gamma(a) = \sum_{n=1}^{+\infty} c_n a^n.
\end{equation*}
{The asymptotic behavior of the coefficients $c_n$ was found in \cite[Equation (18)]{Be10} (a result that was derived again in \cite{Be10a},} with more orders computed in \cite{BeCl13}){. The asymptotic behavior of the coefficients $c_n$ is}:
\begin{equation} \label{eq:asymp_behavior_cn}
 c_{n+1} = -(3n+2+\mathcal{O}(n^{-1}))c_n.
\end{equation}
Furthermore, one easily check that the first terms of this expansion are given by $c_1=1$ and $c_2=-2$.

One important result we will build upon is
\begin{claim} \label{thm:resurgence_gamma}
 $\hat\gamma$ is $\Z^*/3$-resurgent.
\end{claim}
\begin{rk} \label{remark:need_for_well_behaved1}
 Since this result has not be written in the literature in the form we give above let us explain how it follows from previous work. Notice that we do not 
 write Theorem for this result, since the articles we will quote belong to the physics literature. We choose not to 
 write a complete proof of this result as it would add another rather lengthy and technical section of this article. Such a section, aimed only at making 
 {rigorous} an already known result was not deemed important enough.
 
 In \cite{BeCl13} perturbations to the asymptotic behaviour of $\gamma$ were computed. These were computed using ``formal parameters'' around which 
 perturbation theory was performed. In a more modern language, these formal parameters would be understood as transseries terms, and the expansions 
 computed in \cite{BeCl13} the perturbative {expansions} in the one instanton and one anti-instanton sectors. The {nowadays} well-understood fact that 
 this transseries expansion is equivalent to {computing} around the first poles of the Borel transform $\hat\gamma$ was {explicitly} shown in \cite{BeCl15}.
 
 This implies that the Borel transform $\hat\gamma$ can be continued along paths going further than the convergence disc of radius $1/3$ provided they 
 avoid the two singularities in $\zeta=\pm1/3$. From there, one can choose to work with the Borel transform $\hat\gamma$ or with the formal series 
 $\gamma$. The former was done in \cite{BeCl14}, where it was shown that $\hat\gamma$ has its singularities located on $\Z^*/3$. It can be done with the 
 transseries $\gamma$ using the same technics used in \cite{BeCl13}\footnote{This was not published since the perturbation around higher singularities of the 
 Borel transform are subdominant and difficult to approach numerically.}.
 
 {Finally, let us point out that the singularities of $\hat\gamma$ were computed in \cite[Sections 3.2 and 3.3]{BeCl14}. The precise formula are not important here, but it is to notice that these computations imply that the alien derivatives $\Delta_{\omega}$ acts non-trivially on $\hat\gamma$ for any $\omega\in\Z^*/3$.}
\end{rk}
Since we want to resum the two-point function in the direction $\theta=0$, we will focus on this direction. Therefore we will set $\Omega=\N^*/3$
and move on to prove that the two-point function is $\Omega$-resurgent.

\section{Resurgent analysis of the RGE}

\subsection{Solution of the renormalisation group equation}

We want to study the two-points function $G(L,a)$ as a formal series in $a$. We first show that $G(L,a)$ is indeed such a formal series thanks 
to the following lemma.
\begin{lem}
 For any $L\in\C$; the formula \eqref{eq:L_expansion_G} defines a formal series in $a$ with coefficients depending on $L$.
\end{lem}
\begin{proof}
 Since $\gamma_0(a)=1$ by definition and $\gamma_1(a)=\gamma(a)$ lies in $a\C[[a]]$ as a result of \cite{BeLoSc07}, we obtain from
 \eqref{eq:recursion_gamma} with a trivial induction that, for any $k\in\N$, $\gamma_k(a)\in a^k\C[[a]]$.
 Then, for $n\geq1$, contributions to $a^n$ in $G(L,a)$ can only come from $\gamma_1(a),\cdots,\gamma_n(a)$ and their sum is therefore finite.
\end{proof}
The fact that we had to make this small manipulation indicates that the expansion \eqref{eq:L_expansion_G} is not suited to the study of $G(L,a)$ as a
formal series in $a$. We will therefore use the following alternative expansion of the two-points function.
\begin{equation} \label{eq:G_a_exp}
 G(L,a) = \sum_{n=0}^{+\infty} g_n(L)a^n\in A[[a]]
\end{equation}
with $A$ some suitable algebra of smooth functions or formal series.
\begin{prop} \label{prop:solution_RGE}
 The renormalisation group equation \eqref{eq:RGE} admits a solution of the form \eqref{eq:G_a_exp}, with $A=\C[L]$, explicitly given by $g_0(L)=1$ and
 \begin{equation} \label{eq:ansatz_sol_RGE}
  g_n(L) = \sum_{q=1}^n\left(\sum_{\substack{i_1,\cdots,i_q>0\\i_1+\cdots+i_q=n}}c_{i_1}\cdots c_{i_q}K_{i_1\cdots i_q}\right)\frac{L^q}{q}
 \end{equation}
 with the $c_n$ the coefficients of $\gamma(a)$ and $K_{i_1\cdots i_q}$ real numbers inductively defined for any $n\in\N$ and $q\in\{2,\cdots,n+1\}$ by $K_{n}=1$ and
 \begin{equation*}
  K_{i_1\cdots i_q} = (1+3(n+1-i_q))K_{i_1\cdots i_{q-1}}
 \end{equation*}
 with $i_1+\cdots +i_q=n+1$.
\end{prop}
\begin{proof}
 First, observe that the SDE \eqref{eq:SDnlin} taken at $a=0$ gives $G(L,0)=g_0(L)=1$. Furthermore, the RGE \eqref{eq:RGE} implies the following 
 family of differential equations (with $n\geq1$) when one replaces $G(L,a)$ by its representation \eqref{eq:G_a_exp}
 \begin{equation*}
  g_n'(L) = \sum_{p=1}^n c_p(1+3(n-p))g_{n-p}(L).
 \end{equation*}
 Notice that at this stage the derivative can be the derivative of function or the derivative of formal series.
 
 We now prove that these equations are solved as claimed by \eqref{eq:ansatz_sol_RGE} by induction. For the case $n=1$, the equation reduces to
 $g_1'(L)=1$ since $c_1=1=g_0(L)$. This is solved to $g_1(L)=L$ since by the expansion \eqref{eq:G_a_exp}, $G(L,a)$ has only $1=g_0(L)$ as a term 
 independent of $L$. We thus find $K_1=1$ as claimed. 
 
 It will be important for the induction step to have performed the case $n=2$. Observing that $g_1(L)=c_1L$ since $c_1=1$ we find for $g_2$ the 
 equation $g'_2(L) = c_1(1+3(2-1))c_1L + c_2$. This integrates to 
 \begin{equation*}
  g_2(L) = (c_1)^2(1+3(2-1))\frac{L^2}{2} + c_2L
 \end{equation*}
 without constant term for the same reason than the case $n=1$ treated above. We then find $K_2=1$ and $K_{11}=(1+3(2-1))\frac{K_1}{2}$ 
 as claimed.
 
 Let us now assume that the statement of the proposition holds for $n\geq2$. Writing aside the term $p=n+1$, integrating and switching the sum over 
 $q$ by one we find
 \begin{equation*}
  g_{n+1}(L) = c_{n+1}L + \sum_{p=1}^nc_p(1+3(n+1-p)\sum_{q=2}^{n+1-(p-1)}\frac{L^q}{q} \sum_{\substack{i_1,\cdots,i_{q-1}>0\\i_1+\cdots+i_{q-1}=n+1-p}}c_{i_1}\cdots c_{i_{q-1}}K_{i_1\cdots i_{q-1}}.
 \end{equation*}
 As before, we do not have a constant term thanks to the expansion \eqref{eq:G_a_exp}.
 
 Noticing that $\sum_{p=1}^n\sum_{q=2}^{n+1-(p-1)}=\sum_{q=2}^{n+1}\sum_{p=1}^{n+1-(q-1)}$ we can rewrite $g_{n+1}(L)$ as 
 \begin{equation*}
  c_{n+1}L + \sum_{q=2}^{n+1}\frac{L^q}{q}\sum_{p=1}^{n+1-(q-1)}(\cdots).
 \end{equation*}
 Now we can relabel the sum over $p$ as a sum over $i_q$. Thus the sums over $p$ and $i_1,\cdots,i_{q-1}$ can be merged. We obtain
 \begin{equation*}
  g_{n+1}(L) = c_{n+1}L +  \sum_{q=2}^{n+1}\frac{L^q}{q}\left(\sum_{\substack{i_1,\cdots,i_{q}>0\\i_1+\cdots+i_{q}=n+1}}c_{i_1}\cdots c_{i_q}\underbrace{(1+3(n+1-i_q))K_{i_1\cdots i_{q-1}}}_{=:K_{i_1\cdots i_q}}\right)
 \end{equation*}
 We therefore have the right form for $g_{n+1}(L)$, $K_{n+1}=1$ and the induction relation over the $K_{i_1\cdots i_q}$ claimed in the 
 Proposition.
\end{proof}

\subsection{The two-point function is 1-Gevrey}

To prove that the formal series \eqref{eq:G_a_exp} is indeed 1-Gevrey, we first need a reformulation of the formula \eqref{eq:asymp_behavior_cn}.
\begin{lem} \label{lem:bounds_cn}
 For any $n\in\N^*$, the following bounds hold
 \begin{equation*}
  (3\delta)^{n-1}(n-1)!\leq |c_n| \leq (3K)^n n!
 \end{equation*}
 for some $K>1$ and $\delta\in]0,1]$.
\end{lem}
\begin{proof}
 The proof is by induction. The case $n=1$ holds since $c_1=1$. Assuming both inequalities hold for $n\in\N^*$, we first have
 \begin{equation*}
  |c_{n+1}| = |3n+2+\mathcal{O}(n^{-1})||c_n| \leq 3K(n+1)|c_n|
 \end{equation*}
 provided $K$ has been chosen large enough. The upper bound of $|c_{n+1}|$ then follows from the upper bound of $|c_n|$. For the lower bound, one 
 writes
 \begin{equation*}
  |c_{n+1}| = |3n+2+\mathcal{O}(n^{-1})||c_n| 
  \geq 3n\delta|c_n|
 \end{equation*}
 (provided $\delta$ has been chosen small enough) and the lower bound of $|c_{n+1}|$ then follows from the lower bound of $|c_n|$.
\end{proof}
One can {easily show} that 
\begin{equation*}
 \frac{1}{q} K_{i_1\cdots i_q} \leq \frac{1}{n}K_{\underbrace{1\cdots1}_{n\text{ times}}} = (3n-2)!!!
\end{equation*}
with $n=i_1+\cdots+i_q$ and $(3n-2)!!!=\prod_{i=0}^{n-1}(3n-2-i)$.
However this bound is too crude: we need a bound that is not uniform in $q$. Indeed, one obtain from the Lemma \ref{lem:bounds_cn} that the term 
$c_{i_1}\cdots c_{i_q}$ in the solution \eqref{eq:ansatz_sol_RGE} is dominated by the case $q=1$ while the term $K_{i_1\cdots i_q}$ is dominated by 
the term $q=n$. It is the fact that these two bounds cannot be reached together that will allow to prove that the solution \eqref{eq:ansatz_sol_RGE} is 
1-Gevrey.

Recall that for $n\in\N^*$, a {\bf composition} of $n$ is a finite sequence $(i_1,\cdots,i_q)$ of strictly positive integers such that ${i_1}+\cdots+i_q=n$. 
For any composition $(i_1,\cdots,i_q)$ of $n\in\N^*$ recall that the 
{\bf multinomial number} $\binom{n}{i_1,\cdots,i_q}$ is defined by
\begin{equation*}
 \binom{n}{i_1,\cdots,i_q} := \frac{n!}{i_1!\cdots i_q!}.
\end{equation*}
These numbers famously appear in the multinomial theorem and have many important combinatorics properties.
\begin{lem} \label{lem:bound_K_w}
 For any $n$ in $\N^*$ and composition $(i_1,\cdots,i_q)$ of $n$, we have
 \begin{equation*}
  \frac{1}{q}K_{i_1\cdots i_q} \leq \frac{3^n}{n}\binom{n}{i_1,\cdots,i_q}.
 \end{equation*}
\end{lem}
\begin{proof}
 First, observe that, for any $n\in\N^*$, the case $q=1$ trivially hold since $K_n=1=\binom{n}{n}$. We now prove that the result holds for every $n$ and 
 every $q$ by induction over $n$.
 
 For $n=1$, the inequality trivially holds (it is the equality case). Assume it holds for all $p\in\{1,\cdots,n\}$ for some $n\in\N^*$ and let 
 $(i_1,\cdots,i_q)$ be a composition of $n+1$. We have 
 already seen {that} if $q=1$ the result holds. If $q\geq2$ we then have
 \begin{equation*}
  \frac{1}{q}K_{i_1\cdots i_q} 
  \leq (1+3(n+1-i_q)\frac{K_{i_1\cdots i_{q-1}}}{q-1} \leq (1+3(n+1-i_q)\frac{3^{n+1-i_q}}{n+1-i_q}\binom{n+1-i_q}{i_1,\cdots,i_{q-1}}
 \end{equation*}
 by the induction hypothesis, which we can use since $q\geq2$ and thus $i_q\in\{1,\cdots,n\}$.
 
 From the definition of the multinomial numbers, we have
 \begin{equation*}
  \binom{n+1-i_q}{i_1,\cdots,i_{q-1}} = \binom{n+1}{i_q}^{-1}\binom{n+1}{i_1,\cdots,i_q}.
 \end{equation*}
 The result on rank $n+1$ then follows from the observation that 
 \begin{equation*}
  \left(3+\frac{1}{n+1-i_q}\right)\binom{n+1}{i_q}^{-1} \leq 3^{i_q}
 \end{equation*}
 for every $n\in\N^*$ and $i_q\in\{1,\cdots,n\}$.
\end{proof}
We are now ready to prove the main result of this subsection, namely that the two-point function is 1-Gevrey
\begin{prop} \label{prop:G_one_Gevrey}
 The two-point function $G(L,a)$ is 1-Gevrey as a formal series in $a$: for any $L\in\R$ 
 \begin{equation*}
  |g_n(L)| \leq \frac{3}{2}(18K^2 \tilde L)^n n!
 \end{equation*}
 with $\tilde L:=\max\{L,1\}$ and $K$ the constant appearing in the upper bound of $|c_n|$ in Lemma \ref{lem:bounds_cn}.
\end{prop}
\begin{rk}
 In practice, we are interested in the non perturbative regime which in the WZ model appears for $p^2=\mu^2\exp(L)\to\infty$. In this regime, we see 
 that the locus of the first singularity of the two-point function could depend on $L$ and in particular go to zero as $L\to\infty$. We will see later 
 that this is not the case. However the first singularities of $\hat{G}(\zeta,L)$ can move in an intermediate regime.
 This indicates that the singularities of the Borel transform\footnote{at least the first one, but since a singularities in $\omega\in\C^*$
 generally produces new singularities in $\omega\N^*$ (as in Example \ref{ex:comvo_resu}), we expect that all singularities will depend on $L$, at 
 least in some 
 non perturbative regime.} contains non perturbative information of the theory (which is not a new 
 observation: see for example \cite{Ho79}). Therefore resurgence theory has to be an important tool to unravel 
 non perturbative aspects of QFTs.
\end{rk}
\begin{proof}
 Using Lemma \ref{lem:bounds_cn} we have
\begin{equation*}
  \left|\frac{c_n}{c_{i_1}\cdots c_{i_q}}\right| \geq \frac{(3\delta)^{n-1}(n-1)!}{(3K)^{i_1}i_1!\cdots(3K)^{i_q}i_q!}  = \frac{1}{3n}\frac{\delta^{n-1}}{K^n}\binom{n}{i_1,\cdots,i_q} = \frac{1}{3n}\frac{1}{K^n}\binom{n}{i_1,\cdots,i_q}.
 \end{equation*}
 Using this as an upper bound for $|c_{i_1}\cdots c_{i_q}|$ together with the bound for $\frac{1}{q}K_{i_1\cdots i_q}$ of Lemma 
 \ref{lem:bound_K_w} we obtain
 \begin{equation} \label{eq:bound_gn}
  |g_n(L)| \leq 3\sum_{q=1}^n\left(\sum_{\substack{i_1,\cdots,i_q>0 \\ i_1+\cdots+i_q=n}}(3K)^n|c_n|\right)L^q  = 3 (3K)^n|c_n| \sum_{q=1}^n\binom{n-1}{q-1}L^q
 \end{equation}
 where we have used the simple combinatorial result that there are $\binom{n-1}{q-1}$ compositions of $n$ with 
 length $q$. Using that $L^q\leq \tilde L^n$ for any $q\in\{1,\cdots,n\}$ and once more the upper bound for $|c_n|$ of Lemma \ref{lem:bounds_cn} we find 
 the result of the Theorem since $\sum_{q=1}^n\binom{n-1}{q-1}=2^{n-1}$.
\end{proof}
\begin{rk}
 One can use the bound \eqref{eq:bound_gn} more directly to find a more precise bound:
 \begin{equation*}
  |g_n(L)| \leq 3 (9K^2)^n L(L+1)^{n-1}n!
 \end{equation*}
 which holds for all $L$. This bound indicates that the first singularities of the Borel transform is rejected to infinity in the perturbative 
 limit $L\to0$ (but \emph{not} that $G(L,a)$ is analytic in this limit), and therefore that the non perturbative effects encoded in the singularities 
 of the Borel transform vanish as expected in the perturbative limit $L\to 0$.
\end{rk}

\subsection{The two-point function is resurgent} \label{subsec:G_res}

We start with an easy Lemma:
\begin{lem} \label{lem:gamma_n_resu}
 The function $\hat\gamma_n$ is  $\Omega$-resurgent for all $n$ in $\N^*$.
\end{lem}
\begin{proof}
 This result is a direct consequence of the fact that the space of $\Omega$-resurgent functions is stable under convolution, derivation and
 multiplication by an analytic function together with the fact that $\hat\gamma$ is resurgent ({Claim} \ref{thm:resurgence_gamma}). This Lemma 
 is {then easily} shown by induction using the renormalisation group equation \eqref{eq:RGE_borel}.
\end{proof}
The space of resurgent functions is stable by sums, but the above Lemma is not enough to prove that 
$\sum_{n\geq1}\hat\gamma_n(\zeta)\frac{\Lambda^n}{n!}=:\hat{G}(\zeta,\Lambda)$ is $\Omega$-resurgent. In order to tame the combinatorics of the objects appearing in 
the proof, let us introduce some intermediate objects.
\begin{defn}
 For any $n\in\N^*$ define the set $W_n$ as the subset of words written in the alphabet $\{\star,.\}$ such that 
 \begin{equation*}
  W_1:=\{\emptyset\},\quad W_{n+1}:=\{(\star)\sqcup w|w\in W_n\}\bigcup\{(\star.)\sqcup w|w\in W_n\}
 \end{equation*}
 with $\sqcup$ the concatenation product of words. We further set $W:=\bigcup_{n\in\N^*} W_n$.
\end{defn}
\begin{lem} \label{lem:W_n}
 For any $n\in\N^*$ we have $|W_n|=2^{n-1}$.
\end{lem}
\begin{proof}
 For any $n\in\N^*$ write $W_{n+1}=A_n\bigcup B_n$ with $A_n:=\{(\star)\sqcup w|w\in W_n\}$ and $B_n:=\{(\star.)\sqcup w|w\in W_n\}$. Let us check 
 that $A_n\cap B_n=\emptyset$. Let $W_{n+1}\ni w\in A_n\cap B_n$. Then it exists $w_1\in A_n$ and $w_2\in B_n$ such that 
 \begin{equation*}
  w = (\star)\sqcup w_1=(\star.)\sqcup w_2.
 \end{equation*}
 This implies that $w_1\neq\emptyset$ and since every nonempty word in $W$ starts with $\star$ we can write $w_1=(\star)\sqcup w_3$ for some word $w_3$ not 
 necessarily in $W$. We then have $w=(\star\star)\sqcup w_3 = (\star.)\sqcup w_2$ which a contradiction. Then $A_n\cap B_n=\emptyset$ and 
 $|W_{n+1}|=2|W_n|$. The result then follows from $|W_1|=1=2^0$.
\end{proof}
Finally, let us prove a simple but useful lemma about analytic continuation of series.
\begin{lem} \label{lem:analytic_continuation_series}
 Let $U\subset V$ be two open subsets of $\C$. Let $f_n:U\mapsto \C$ be a sequence of holomorphic functions such that:
 \begin{enumerate}
  \item $f:=\sum_{n=0}^\infty f_n$ is holomorphic in $U$;
  \item $f_n$ admits an analytic continuation $\tilde f_n$ to $V$;
  \item $\tilde f_n$ is bounded on $V$ by an analytic function $F_n$: $|\tilde f_n|\leq F_n$;
  \item The series $F=\sum_{n=0}^\infty F_n$ converges in $V$.
 \end{enumerate}
 Then $f$ admits an analytic continuation $\tilde f$ to $V$ and $|\tilde f|\leq F$.
\end{lem}
\begin{proof}
 For any $z\in V$, let us set 
 \begin{equation*}
  S_N(z) := \sum_{n=0}^N|\tilde f_n(z)| \leq \sum_{n=0}^N F_n(z) \longrightarrow F(z)
 \end{equation*}
 as ${N\to \infty}$. Then $S_N(z)$ is increasing and bounded and therefore convergent. {Hence} the series $\tilde f(z) := \sum_{n=0}^\infty\tilde f_n(z)$ is absolutely 
 convergent and thus convergent. This series {is by definition} an analytic continuation of $f$ to $V$ and is bounded by $F$.
\end{proof}
We are now ready to prove the main result of this section.
\begin{thm} \label{thm:resurgence_two_points_function}
 For any $\Lambda\in\R$, the map $\zeta\mapsto\hat{G}(\zeta,\Lambda)$ is $\Omega$-resurgent. 
\end{thm}
\begin{proof}
 Let $\delta,L>0$ with $\delta<\rho(\Omega)/2$. Let $\gamma$ be a path in $\calK_{\delta,L}(\Omega)$.
 According to Lemma \ref{lem:analytic_continuation_series} we only need to prove that the series 
 \begin{equation*}
  \sum_{n\geq1}(\text{cont}_\gamma\hat\gamma_n)(\zeta)\frac{\Lambda^n}{n!}
 \end{equation*}
 converges normally. Indeed, in this case, it will be equal to
 \begin{equation*}
  (\text{cont}_\gamma\hat G)(\zeta,\Lambda):=\left(\text{cont}_\gamma\sum_{n=1}^{\infty}\hat\gamma_n\right)(\zeta).
 \end{equation*}
For any $N\in\N^*$, we will deduce from a bound on $\hat\gamma$ a bound on $\hat\gamma_{N+1}$ in the domain $\calK_{\delta,L}(\Omega)$ which contain 
 the path $\gamma$. So, fix $N\in\N^*$ and for $n\in\{1,\cdots,N+1\}$, set 
 \begin{equation*}
  \delta_n:=\frac{\delta}{2} + (n-1)\frac{\delta}{2N},\quad L_n:=L+\frac{\delta}{2} - (n-1)\frac{\delta}{2N}.
 \end{equation*}
 We did not write the dependence on $N$ of $\delta_n$ and $L_n$ to lighten the notations. Notice however that $\delta_1=\delta/2$ and $L_1=L+\delta/2$ for 
 any $N\in\N^*$.

 We now define a map 
 \begin{align*}
  f:W & \longrightarrow \widehat{\mathcal{R}}_\Omega \\
    w & \longmapsto f_w
 \end{align*}
 recursively by
 \begin{equation*}
  f_\emptyset(\zeta) :=  |\hat\gamma(\zeta)| + S,\quad f_{(\star)\sqcup w}(\zeta):=4(f_\emptyset\star f_w)(\zeta),\quad f_{(\star.)\sqcup w}(\zeta):=\frac{6NK}{\delta}(f_\emptyset\star f_w)(\zeta)
 \end{equation*}
 where we have set 
 \begin{equation*}
  S:=\max_{\zeta\in\calK_{\delta_1,L_1}(\Omega)}|\hat\gamma(\zeta)|\quad\text{and}\quad K:= \max_{\zeta\in\calK_{\delta_1,L_1}(\Omega)}|\zeta|.
 \end{equation*}
 The map $f$ is well-defined due to the proof above that the sets $A_n$ and $B_n$ do not intersect. Furthermore its image is a subset 
 of the $\Omega$-resurgent functions since they are stable by convolution and by multiplication by analytic functions.
 
 The analytical part of this proof is now essentially contained is the next Lemma.
 \begin{lem} \label{lem:analytical_bound_f_w}
  For any $N\in\N^*$ and $n\in\{1,\cdots,N+1\}$ we have 
  \begin{equation*}
   |\hat\gamma_n(\zeta)|\leq \sum_{w\in W_n}f_w(\eta)
  \end{equation*}
  for any $\zeta,\eta\in\calK_{\delta_n,L_n}(\Omega)$.
 \end{lem}
 \begin{proof}
  We prove this result by induction on $n$. For $n=1$ we have $f_\emptyset(\zeta)\geq S=\max_{\zeta\in\calK_{\delta_1,L_1}(\Omega)}|\hat\gamma(\zeta)|$ 
  and therefore the 
  lemma holds. Assume it holds for $n\in\{1,\cdots,N\}$. We then have, for any $\zeta\in\calK_{\delta_{n+1},L_{n+1}}(\Omega)$
  \begin{equation*}
   |\hat\gamma_{n+1}(\zeta)| \leq 4|\hat\gamma|\star|\hat\gamma_n|(\zeta) + 3|\hat\gamma|\star|\zeta\partial_\zeta\hat\gamma_n|(\zeta).
  \end{equation*}
  Then using the induction hypothesis and the continuity of the convolution product we have
  \begin{equation*}
   4|\hat\gamma|\star|\hat\gamma_n|(\zeta) \leq \sum_{w\in W_n} 4(f_\emptyset\star f_w)(\eta) = \sum_{w\in W_n}f_{(\star)\sqcup w}(\eta)
  \end{equation*}
  for any $\eta\in\calK_{\delta_{n},L_{n}}(\Omega)\subset\calK_{\delta_{n+1},L_{n+1}}(\Omega)$.
  
  Now, by definition, for any $\zeta\in\calK_{\delta_{n+1},L_{n+1}}(\Omega)$, the disc of center $\zeta$ and radius $\frac{\delta}{2N}$ lies in 
  $\calK_{\delta_{n},L_{n}}(\Omega)$. Therefore, using the definition of $K$ and Cauchy inequality on the disc of center $\zeta$ and radius 
  $\frac{\delta}{2N}$ we find
  \begin{equation*}
   |\zeta\partial_\zeta\hat\gamma_n(\zeta)| \leq \frac{2NK}{\delta}\max_{\zeta\in D(\zeta,\delta/2N)}|\hat\gamma_n(\zeta)| \leq \sum_{w\in W_n}\frac{2NK}{\delta}f_w(\eta)
  \end{equation*}
  for any $\eta\in\calK_{\delta_{n},L_{n}}(\Omega)\subset\calK_{\delta_{n+1},L_{n+1}}(\Omega)$. Thus
  \begin{equation*}
   3|\hat\gamma|\star|\zeta\partial_\zeta\hat\gamma_n|(\zeta) \leq \sum_{w\in W_n}\frac{6NK}{\delta}(f_\emptyset\star f_w)(\eta) = f_{(\star.)\sqcup w}(\eta)
  \end{equation*}
  for any $\eta\in\calK_{\delta_{n},L_{n}}(\Omega)\subset\calK_{\delta_{n+1},L_{n+1}}(\Omega)$. Combining this bound with the one for $4|\hat\gamma|\star|\hat\gamma_n|(\zeta)$
  we obtain
  \begin{equation*}
   |\hat\gamma_{n+1}(\zeta)| \leq \sum_{w\in W_n}\left(f_{(\star)\sqcup w}(\eta) + f_{(\star.)\sqcup w}(\eta)\right) = \sum_{w\in W_{n+1}}f_{w}(\eta)
  \end{equation*}
  for any $\eta\in\calK_{\delta_{n+1},L_{n+1}}(\Omega)$.
 \end{proof}
 We now need to bound $f_w$. 
 {Let $||w||$ be} the number of times the 
 letter $.$ is present in the word $w\in W$. Then for any $n\in\{1,\cdots,N+1\}$ and $w\in W_n$ we have
 \begin{equation*}
  f_w(\zeta) = \left(\frac{6NK}{\delta}\right)^{||w||}4^{n-||w||}f_\emptyset^{\star n}(\zeta).
 \end{equation*}
 We can now use Sauzin's bound \eqref{eq:bound_conv_resu} for $n=N+1$:
 \begin{equation*}
  \max_{\zeta\in\calK_{\delta,L}(\Omega)} f_w(\zeta) \leq \left(\frac{6NK}{\delta}\right)^N4^{N+1}\frac{C^{N+1}}{(N+1)!}\left(\max_{\zeta\in\calK_{\delta/2,L+\delta/2}(\Omega)}f_\emptyset(\zeta)\right)^{N+1}
 \end{equation*}
 where we have used that $||w||\in\{0,1,\cdots,N\}$. Now, using that $\delta/2=\delta_1$ and $L+\delta/2=L_1$ we find 
 $\max_{\zeta\in\calK_{\delta/2,L+\delta/2}(\Omega)}f_\emptyset(\zeta)=2S$. Using Lemmas \ref{lem:analytical_bound_f_w} and \ref{lem:W_n} we obtain
 \begin{equation*}
  \max_{\zeta\in\calK_{\delta,L}(\Omega)}|\hat\gamma_{N+1}(\zeta)| \leq \frac{\delta}{12K}\left(\frac{96}{\delta}SKC\right)^{N+1}\frac{N^N}{(N+1)!}.
 \end{equation*}
 Using Stirling's formula we then have the following bound, for $N$ big 
 \begin{equation*}
  \max_{\zeta\in\calK_{\delta,L}(\Omega)}|\hat\gamma_{N+1}(\zeta)| \leq \frac{\delta}{12Ke}\left(\frac{96}{\delta}SKCe\right)^{N+1} \frac{1}{\sqrt{2\pi N}N}\left(1+\mathcal{O}\left(\frac{1}{\sqrt{N}}\right)\right).
 \end{equation*}
 This implies the normal convergence of the series 
 $\sum_{n\geq1}(\text{cont}_\gamma\hat\gamma_n)(\zeta)\frac{\Lambda^n}{n!}=:(\text{cont}_\gamma\hat G)(\zeta,\Lambda)$ and concludes the proof.
\end{proof} 
{ 
\begin{rk}
 From Remark \ref{remark:need_for_well_behaved1} and Lemma \ref{lem:gamma_n_resu} one obtains that an infinite number of alien derivative act non-trivially on each of the $\hat\gamma_n$. As a consequence, and excluding miraculous cancellation of singularities, the same is true for $\zeta\mapsto\hat G(\zeta,\Lambda)$. This is {corroborated} by the computations of \cite[Section 4.1]{BeCl16} where the main contributions to the (lateral) alien derivatives applied to $\hat G$ were computed and shown to be non-zero. Therefore, the full theory or well-behaved averages is needed for the summation of the two-points function of the Wess-Zumino model.
\end{rk}
}
 
\section{Asymptotic bound of the two-point function}

We now prove that $\hat G(\zeta,L)$ admits an exponential bound in the {star-shaped domain} $\mathcal{U}_\Omega$ of $\C\dsetminus\Omega$ {introduced in Section \ref{subsec:res_fct}}.

\subsection{Statement of the problem}

The following lemma implies that one actually {needs} to study the Schwinger-Dyson equation in order to find the right type of bound on the two-point function.
\begin{lem} \label{lem:bound_gamma_n}
 Let $g:\mathcal{U}_\Omega\longrightarrow\R_+$ be an increasing analytic function such that, for any $\zeta\in\mathcal{U}_\Omega$
 \begin{equation*}
  \max\Big\{\max_{\eta\in[0,\zeta]}|\hat\gamma(\eta)|,\max_{\eta\in[0,\zeta]}|\hat\gamma'(\eta)|\Big\} \leq g(\zeta).
 \end{equation*}
 Then for any $n\in\N^*$ we have
 \begin{equation*}
  \max\Big\{\max_{\eta\in[0,\zeta]}|\hat\gamma_n(\eta)|,\max_{\eta\in[0,\zeta]}|\hat\gamma_n'(\eta)|\Big\} \leq \left[(4+3|\zeta|)(1 + g(\zeta)|\zeta|)\right]^{n-1} g(\zeta).
 \end{equation*}
\end{lem}
\begin{rk}
 The function $g$ exists since $\hat\gamma$ and $\hat\gamma'$ are analytic (but not bounded) on $\mathcal{U}_\Omega$.
\end{rk}
\begin{proof}
 We prove this Lemma by induction. The case $n=1$ holds by definition of $g$. Assuming the Lemma holds for some $n\in\N^*$; we use the bound 
 \eqref{eq:bound_star_shaped} (which we can use on $\mathcal{U}_\Omega$ since it is star-shaped with respect to the origin) 
 on the renormalisation group equation \eqref{eq:RGE_borel} to obtain, for any $\zeta\in\mathcal{U}_\Omega$
 \begin{align*}
  |\hat\gamma_{n+1}(\zeta)| & \leq g(\zeta)|\zeta|(4\max_{\eta\in[0,\zeta]}|\hat\gamma_n(\eta)|+3|\zeta|\max_{\eta\in[0,\zeta]}|\hat\gamma'_n(\eta)|) \\
			    & \leq (4+3|\zeta|)g(\zeta)|\zeta|\max\Big\{\max_{\eta\in[0,\zeta]}|\hat\gamma_n(\eta)|,\max_{\eta\in[0,\zeta]}|\hat\gamma_n'(\eta)|\Big\} \\
			    & \leq (4+3|\zeta|)(1+g(\zeta)|\zeta|)\max\Big\{\max_{\eta\in[0,\zeta]}|\hat\gamma(\eta)|,\max_{\eta\in[0,\zeta]}|\hat\gamma'(\eta)|\Big\}.
 \end{align*}
 For any $\eta\in[0,\zeta]$ we further have
 \begin{align*}
  |\hat\gamma_{n+1}(\eta)| & \leq (4+3|\eta|)(1+g(\eta)|\eta|)\max\Big\{\max_{\sigma\in[0,\eta]}|\hat\gamma_n(\sigma)|,\max_{\sigma\in[0,\eta]}|\hat\gamma_n'(\sigma)|\Big\} \\
			    & \leq (4+3|\zeta|)(1+g(\zeta)|\zeta|)\max\Big\{\max_{\eta\in[0,\zeta]}|\hat\gamma(\eta)|,\max_{\eta\in[0,\zeta]}|\hat\gamma'(\eta)|\Big\}
 \end{align*}
 since we have assumed $g$ to be increasing.
 Therefore $\max_{\eta\in[0,\zeta]}|\hat\gamma_{n+1}(\zeta)|$ admits the bound of the Lemma.
 
 To obtain a bound on $|\hat\gamma_{n+1}'(\zeta)|$ we use Leibniz's formula
 \begin{equation} \label{eq:Leibniz}
  \frac{d}{dt}\int_{a(t)}^{b(t)} f(t,x)dx = b'(t)f(t,b(t)) - a'(t)f(t,a(t)) + \int_{a(t)}^{b(t)} \frac{\partial f}{\partial t}(t,x)dx;
 \end{equation}
 which holds provided $a$, $b$  and $f$ are $\mathscr{C}^1$.
 
 In our case this formula gives
 \begin{equation*}
  \partial_\zeta(f\star g)(\zeta) = f(0)g(\zeta) + \int_0^\zeta f'(\zeta-\eta)g(\eta) d\eta = f(\zeta)g(0) + \int_0^\zeta f(\zeta-\eta) g'(\eta) d\eta.
 \end{equation*}
 (one gets the second equality through an integration by part). Using $\hat\gamma(0)=1$ and again the bound \eqref{eq:bound_star_shaped} {on} the
 renormalisation group equation \eqref{eq:RGE_borel} {derived once one obtains}, for any $\zeta\in\mathcal{U}_\Omega$
 \begin{equation*}
  |\hat\gamma_{n+1}'(\zeta)|\leq \left[(4+3|\zeta|)(1 + g(\zeta)|\zeta|)\right]\max\Big\{\max_{\eta\in[0,\zeta]}|\hat\gamma_n(\eta)|,\max_{\eta\in[0,\zeta]}|\hat\gamma_n'(\eta)|\Big\}
 \end{equation*}
 The same bound holds for any $\eta\in[0,\zeta]$ from the same argument than the one used for $\hat\gamma_n$.
 
 From these bounds, the Lemma holds by induction.
 \end{proof}
 {Summing} these $\hat\gamma_n$ we end up with the following 
bound for the two-points function (at infinity):
\begin{equation*}
 |\hat G(\zeta,L)| \leq K\exp(c|\zeta|^2g(\zeta)L),
\end{equation*}
for some bound $g(\zeta)$ of $\hat\gamma$ and $\hat\gamma'$ at infinity. This is too weak a bound to apply Borel-\'Ecalle resummation method. The 
square of $|\zeta|$ comes from the $\zeta$ in the renormalisation group 
equation \eqref{eq:RGE_borel} and the $\zeta^{n-1}$ in the Equation \eqref{eq:bound_star_shaped}, which we used with $n=2$. In order to apply 
Borel-\'Ecalle resummation without accelero-summation, we have two challenges to tackle:
\begin{itemize}
 \item relate the bounds for $\hat\gamma_n$ and for $\hat\gamma'_n$ in order to get ride of one of the power of $\zeta$;
 \item find a specific bound on the asymptotic behavior of $\hat\gamma$.
\end{itemize}
The second issue will be solved using the Schwinger-Dyson equation, but the solution of the first one will actually use inputs from the Schwinger-Dyson 
equation as well.

 \subsection{Rewriting the Schwinger-Dyson equation}
 
Expanding the sum in the Schwinger-Dyson equation in the Borel plane, and using $\mathcal{B}(af(a))=1\star \hat f$ we find 
\begin{equation*}
  \hat\gamma(\zeta) = 1 +2\sum_{n=1}^{+\infty}X_{0n}(1\star\hat\gamma_n)(\zeta) + \sum_{n,m=1}^{+\infty}X_{nm}(1\star\hat\gamma_n\star\hat\gamma_m)(\zeta).
  \end{equation*}
with 
\begin{equation*}
 X_{nm}:=\frac{1}{n!m!}\frac{d^n}{dx^n}\frac{d^m}{dy^m}H(x,y)|_{x=y=0}.
\end{equation*}
Using the representation \eqref{eq:Mellin} of the Mellin transform $H$, we find $X_{0n}=X_{n0}=(-1)^n$. Indeed the series 
$\sum_{k=1}^{+\infty}\frac{\zeta(2k+1)}{2k+1}\left((x+y)^{2k+1}-x^{2k+1}-y^{2k+1}\right)$ contains no terms of the form 
$x^Ny^0$ nor $x^0y^N$. Thus 
\begin{equation*}
 \partial_x^n \left.\exp\Bigl(2\sum_{k=1}^{+\infty}\frac{\zeta(2k+1)}{2k+1}\left((x+y)^{2k+1}-x^{2k+1}-y^{2k+1}\right)\Bigr)\right|_{x=y=0} = 0;
\end{equation*}
and the same holds for the derivatives with respect to $y$. We thus find the Schwinger-Dyson equation in the Borel plane:
\begin{equation} \label{eq:SDE_Borel_expanded}
 \hat\gamma(\zeta) = 1 +2\sum_{n=1}^{+\infty}(-1)^n(1\star\hat\gamma_n)(\zeta) + \sum_{n,m=1}^{+\infty}X_{nm}(1\star\hat\gamma_n\star\hat\gamma_m)(\zeta).
\end{equation}
\begin{rk} \label{rk:analytic_continuation_bound}
 It is crucial to the rest of this proof to realise that, while Equation \eqref{eq:SDE_Borel_expanded} holds for any $\zeta\in\C\dsetminus\Omega$, the 
 series on the R.H.S. only converge in a small open subset of $\C\dsetminus\Omega$ which is mapped to a neighborhood of the origin in $\C$. Indeed, deriving 
 \eqref{eq:SDE_Borel_expanded} we obtain
 \begin{equation*}
  \hat\gamma'(\zeta) = 2\sum_{n=1}^{+\infty}(-1)^n\hat\gamma_n(\zeta) + \sum_{n,m=1}^{+\infty}X_{nm}(\hat\gamma_n\star\hat\gamma_m)(\zeta).
 \end{equation*}
 The renormalisation group equation \eqref{eq:RGE_borel} together with the result of \cite{BeCl14} that $\hat\gamma(\zeta)\sim A\ln(1/3-\zeta)$ when 
 $\zeta$ goes to $1/3$ implies that $\hat\gamma_n$ has the same behavior when $\zeta$ goes to $1/3$. Thus 
 $\sum_{n=1}^{+\infty}(-1)^n\hat\gamma_n(\zeta)$ trivially diverges in an open set close to $1/3$.
 
 Therefore, the series of the R.H.S. of \eqref{eq:SDE_Borel_expanded} should be read as the analytic continuation of these series when one is away from 
 their convergent domain. This will be important since we will use bounds on $\hat\gamma_n$ of the form of  the bounds of Lemma \ref{lem:bound_gamma_n}
 which holds for any $\zeta\in\mathcal{U}_\Omega$. Provided the series of these bounds will admit an analytic extension to the whole of 
 $\mathcal{U}_\Omega$, it will provide a bound for $\hat\gamma$ as needed.
\end{rk}

Now, the other numbers $X_{nm}$ could be computed using the same type of argument we used to find $X_{n0}$, or directly using the Fa\`a-di-Bruno formula. 
However the result of this computation is not particularly enlightening. It will be enough for us to find a bound for $|X_{nm}|$.
\begin{lem} \label{lem:bound_X_nm}
 For any any $r\in]0,1/2[$ it exists a {real positive number}  $K_r>0$ such that, for any $n,m\in\N^*$ we have
 \begin{equation} \label{eq:X_nm}
  |X_{nm}| \leq \frac{K_r}{r^{n+m}}.
 \end{equation}
\end{lem}
\begin{proof}
 We use the multivariate Cauchy inequality (see for example \cite[Theorem 2.2.7]{Hormander66}); namely that if a function $f:\C^n\longrightarrow\C$ is 
 analytic and bounded by $M$ in the polydisc $\{z:|z_i|\leq r_i,~i=1,\cdots,n\}$, then $|\partial^{\alpha}f(0)|\leq M\frac{\alpha!}{r^\alpha}$ for 
 any multi-index $\alpha\in\N^n$ and with obvious notations for factorial and powers. According to \eqref{eq:Mellin0}, the Mellin transform $H$ is 
 analytic in the polydisc $\{(z_1,z_2)\in\C^2:|z_1|\leq r~\wedge~|z_2|\leq r\}$ for any $r\in]0,1/2[$. For any such $r$, set 
 $K_r:=\sup_{|z_1\leq r,z_2\leq r}|H(z_1,z_2)|$. The bound \eqref{eq:X_nm} follows then directly from the multivariate Cauchy inequality.
\end{proof}

\subsection{Intermediate bounds} \label{intermediate}

We start with a common bound of $\hat\gamma$ and $\zeta\partial_\zeta\hat\gamma$ to find bounds on 
$\hat\gamma_n$ and $\hat\gamma_n'$ for any $n\in\N^*$.
\begin{lem} \label{lem:un_autre_lemme}
 Let $g:\mathcal{U}_\Omega\setminus\{0\}\longrightarrow\R$ be a holomorphic function increasing with $|\zeta|$ such that, for any $\zeta\in\mathcal{U}_\Omega\setminus\{0\}$, 
\begin{equation*}
 \max_{\eta\in[0,\zeta]}|\hat\gamma(\eta)|\leq g(\zeta) \quad {\rm and} \quad \max_{\eta\in[0,\zeta]}|\hat\gamma'(\eta)|\leq \frac{g(\zeta)}{|\zeta|}.
\end{equation*}
Let $(g_n)_{n\in\N^*}$ and $(h_n)_{n\in\N^*}$ be two 
 sequences of functions from $\mathcal{U}_\Omega\setminus \{0\}$ 
 to 
 $\R$ inductively defined for any $\zeta\in\mathcal{U}_\Omega\setminus \{0\}$ by $g_1{(\zeta)}:=g(\zeta)$, $h_1{(\zeta)}:=g(\zeta)/|\zeta|$ and 
 \begin{align*}
  g_{n+1}(\zeta) := g(\zeta)|\zeta|\left[4g_n(\zeta)+3|\zeta| h_n(\zeta)\right],\\
  h_{n+1}(\zeta) := \frac{g_{n+1}(\zeta)}{|\zeta|} + 4g_n(\zeta)+3|\zeta|h_n(\zeta).
 \end{align*}
 Then, for any $n\in\N^*$ and $\zeta\in\mathcal{U}_\Omega\setminus \{0\}$
 \begin{equation*}
  \max_{\eta\in[0,\zeta]}|\hat\gamma_n(\eta)| \leq g_n(\zeta), \quad \max_{\eta\in[0,\zeta]}|\hat\gamma_n'(\eta)| \leq h_n(\zeta).
 \end{equation*}
\end{lem}
\begin{rk}
 Such a function $g$ exists since $\hat\gamma$ and $\zeta\partial_\zeta\hat\gamma$ are analytic on $\mathcal{U}_\Omega$. We will later {work with} a 
 bound {that has these properties} but {this latter bound}  will be defined by {a former one} $g$.
\end{rk}
\begin{proof}
 We prove this by induction: the case $n=1$ holds by definition of $g$.
 
 Assuming the result holds for $n\in\N^*$, using the renormalisation group equation \eqref{eq:RGE_borel}, the 
 bound \eqref{eq:bound_star_shaped} and the induction hypothesis we obtain
 \begin{equation*}
  |\hat\gamma_{n+1}(\zeta)| \leq g(\zeta)|\zeta|\left[4g_n(\zeta)+3|\zeta| h_n(\zeta)\right] =: g_{n+1}(\zeta).
 \end{equation*}
 Taking once again the derivative of the renormalisation group equation \eqref{eq:RGE_borel} we obtain, using Leibniz's formula \eqref{eq:Leibniz}
 \begin{equation*}
  \hat\gamma_{n+1}'(\zeta) = 4[\hat\gamma_n(\zeta) + (\hat\gamma'\star\hat\gamma_n)(\zeta)] + 3 [\zeta\hat\gamma_n'(\zeta) + (\hat\gamma'\star(\zeta\hat\gamma_n'))(\zeta)].
 \end{equation*}
 Using the bound \eqref{eq:bound_star_shaped} and the induction hypothesis on this equation gives the result for $\zeta$. The case of 
 $\eta\in[0,\zeta]$ holds from the same argument than the one of Lemma \ref{lem:bound_gamma_n}, which still holds since we assume $g$ to be increasing.
\end{proof}
We can now express together the bounds of $\hat\gamma_n$ and $\hat\gamma_n'$. 
\begin{lem} \label{lem:encore_un_lemme}
 For any $\zeta\in\mathcal{U}_\Omega\setminus \{0\}$, set 
 \begin{equation*}
  \alpha(\zeta) : = \frac{g(\zeta)}{g(\zeta)+1}
 \end{equation*}
 with $g$ a bound of $\hat\gamma$ and $\zeta\hat\gamma'$ as in Lemma \ref{lem:un_autre_lemme}. Then, for any 
 $n\in\N^*$ and any $\zeta\in\mathcal{U}_\Omega\setminus \{0\}$
 \begin{equation*}
  h_n(\zeta) \leq \frac{1}{\alpha(\zeta)}\frac{g_n(\zeta)}{|\zeta|}.
 \end{equation*}
\end{lem}
\begin{proof}
 For $n=1$, the inequality to show is the case $n=1$ of Lemma \ref{lem:un_autre_lemme} since $1/\alpha(\zeta)>1$.
 
 For $n=2$, direct computations give 
 \begin{equation*}
  \frac{1}{\alpha(\zeta)}\frac{g_2(\zeta)}{|\zeta|} = 7g(\zeta)(g(\zeta)+1) \geq h_2(\zeta)=14g(\zeta)
 \end{equation*}
 since $g(\zeta)\geq\max_{\eta\in[0,\zeta]}|\hat\gamma(\zeta)| \geq 1=\hat\gamma(0)$. 
 
 For any $n\geq2$ we have 
 \begin{equation*}
  \frac{1}{\alpha(\zeta)}\frac{g_{n+1}(\zeta)}{|\zeta|} = (g(\zeta)+1)[4g_n(\zeta)+3|\zeta|h_n(\zeta)] = h_{n+1}(\zeta).
 \end{equation*}
 Therefore the result also hold for any $n\geq2$.
\end{proof}
We can now prove the main result of this subsection.
\begin{prop} \label{prop:main_bound}
 Let $g:\mathcal{U}_\Omega\longrightarrow\R$ be a bound of $\hat\gamma$ and $\zeta\hat\gamma'$ as in Lemma 
 \ref{lem:un_autre_lemme}. 
 Then, for any $n\in\N^*$ and $\zeta\in\mathcal{U}_\Omega\setminus \{0\}$ 
 \begin{equation*}
  \max_{\eta\in[0,\zeta]}|\hat\gamma_n(\eta)| \leq \left[(7g(\zeta)+3)|\zeta|\right]^{n-1}g(\zeta).
 \end{equation*}
\end{prop}
\begin{proof}
 By Lemma \ref{lem:un_autre_lemme} it is sufficient to prove $g_n(\zeta)\leq \left[(7g(\zeta)+3)|\zeta|\right]^{n-1}g(\zeta)$ for any $n\in\N^*$. We prove 
 this by induction: the case $n=1$ trivially holds. Assuming the result holds for $n\in\N^*$, we have according to Lemma \ref{lem:encore_un_lemme}
 \begin{equation*}
  g_{n+1}(\zeta) \leq g(\zeta)|\zeta|\left(4+\frac{3}{\alpha(\zeta)}\right)g_n(\zeta) = |\zeta|\left(7g(\zeta)+3)\right)g_n(\zeta)
 \end{equation*}
 by definition of $\alpha(\zeta)$.
\end{proof}

\subsection{Borel-{\'Ecalle} resummation of the two-points function}

The one quantity that we have not bounded yet and that could still give $\hat G$ a {superexponential} behavior at infinity {on the subset
$\mathcal{U}_\Omega$} of $\C\dsetminus\Omega$ is the bound $g$ of $\hat\gamma$. This is taken care of in the next Proposition.
\begin{prop} \label{prop:bound_gamma_infinity}
 {On  $\mathcal{U}_\Omega$,} $|\hat\gamma(\zeta)|$ {and $|\hat\gamma'(\zeta)|$ are} bounded in a neighborhood of infinity by $1$ {and $1/|\zeta|$ respectively}.
\end{prop}
\begin{proof}
 As before
 let $g:\mathcal{U}_\Omega{\setminus\{0\}}\longrightarrow\R$ be a bound of $\hat\gamma$ and $\zeta\hat\gamma'$ as in Lemma 
 \ref{lem:un_autre_lemme}. 
 Using the bound \eqref{eq:bound_star_shaped} on the Schwinger-Dyson equation \eqref{eq:SDE_Borel_expanded}
 with the bounds of Proposition \ref{prop:main_bound} for {the} $\hat\gamma_n$ and the bounds of Lemma \ref{lem:bound_X_nm} for the coefficients $X_{nm}$ we 
 find that $|\hat\gamma|$ is bounded on $\mathcal{U}_\Omega\setminus\{0\}$ by two geometric series. {More properly, and in the spirit of Remark \ref{rk:analytic_continuation_bound},} $|\hat\gamma|$ is bounded in $\mathcal{U}_\Omega\setminus \{0\}$ by the analytic continuation of ({products} of) geometric series. To be more 
 precise, one has 
 \begin{align*}
  |\hat\gamma(\zeta)| & \leq 1 + 2|\zeta|\sum_{n=1}^{\infty} \max_{\eta\in[0,\zeta]}|\hat\gamma_n(\eta)| + \frac{K_r}{2}\max_{\eta\in[0,\zeta]}|\zeta|^2\sum_{n,m=1}^{\infty}\frac{1}{r^{n+m}}\max_{\eta\in[0,\zeta]}|\hat\gamma_n(\zeta)|\max_{\eta\in[0,\zeta]}|\hat\gamma_m(\zeta)| \\
		      & \leq 1 +\frac{2|\zeta| g(\zeta)}{1-(7g(\zeta)+3)|\zeta|} + K_r\left(\frac{|\zeta|g(\zeta)}{r-(7g(\zeta)+3)|\zeta|}\right)^2 =: G(\zeta,g(\zeta))
 \end{align*}
 for any $r\in]0,1/2[$, $\zeta\in\mathcal{U}_\Omega{\setminus\{0\}}$ and with $K_r:=\sup_{|z_1\leq r,z_2\leq r}|H(z_1,z_2)|$. 
 Notice that we removed the $1/2$ in {the third term of} the last 
 bound in order for $G$ to have the following property: for any $\zeta\in\mathcal{U}_\Omega\setminus{\{0\}}$
 \begin{equation} \label{eq:bound_gamma_prime}
  |\hat\gamma'(\zeta)| \leq \frac{G(\zeta,g(\zeta))}{|\zeta|}.
 \end{equation}
 To prove this, we take the derivative of the Schwinger-Dyson equation \eqref{eq:SDE_Borel_expanded}:
\begin{equation*}
 \hat\gamma'(\zeta) = 2\sum_{n=1}^{+\infty}(-1)^n\hat\gamma_n(\zeta) + \sum_{n,m=1}^{+\infty}X_{nm}(\hat\gamma_n\star\hat\gamma_m)(\zeta).
\end{equation*}
Therefore 
\begin{align*}
 |\hat\gamma'(\eta)| & \leq 2\sum_{n=1}^{+\infty}|\hat\gamma_n(\zeta)| + \sum_{n,m=1}^{+\infty}|X_{nm}(\hat\gamma_n\star\hat\gamma_m)(\zeta)| \\
		      & \leq 2\sum_{n=1}^{\infty} \max_{\eta\in[0,\zeta]}|\hat\gamma_n(\eta)| + K_r\max_{\eta\in[0,\zeta]}|\zeta|\sum_{n,m=1}^{\infty}\frac{1}{r^{n+m}}\max_{\eta\in[0,\zeta]}|\hat\gamma_n(\zeta)|\max_{\eta\in[0,\zeta]}|\hat\gamma_m(\zeta)| \\
		      &\leq \frac{G(\zeta,g(\zeta))}{|\zeta|}
\end{align*}
{as claimed.}
It is a cumbersome but simple exercise to study the variations of $G$. However it is enough for the task at hand to check that $G$ is bounded at infinity by $1$. For $\zeta$ {in 
$\mathcal{U}_\Omega$}, we have 
\begin{equation*}
 G(\zeta,X)\sim 1 - \frac{2 X}{7X+3} + K_r\left(\frac{X}{7X+3}\right)^2 =: f(X)
\end{equation*}
for $|\zeta|\to\infty$. We can still choose $r\in]0,1/2[$. Since $H(0,0)=1$ and since $H$ is holomorphic in a neighborhood of $(0,0)$, we can take $r$ small enough 
{for $K_r$ to be} arbitrarily close to $1=H(0,0)$. It then is a simple exercise of real analysis to show that, provided $K_r<7$, $f$ is continuous and decreases over 
 $\R_+^*$. Therefore
 \begin{equation*}
  |\hat\gamma(\zeta)| \lesssim f(0) = 1
 \end{equation*}
 in a neighborhood of infinity. The bound for $\hat\gamma'$ in the same neighborhood of infinity comes from the inequality \eqref{eq:bound_gamma_prime}.
\end{proof}
\begin{thm} \label{thm:bound_two_point_infinity}
 It exists {real} constants $K,~M>0$ such that, for any $L\in\R$, the Borel transform of the solution of the Schwinger-Dyson equation \eqref{eq:SDnlin} and the renormalisation group equation 
 \eqref{eq:RGE} admits the following bound in $\mathcal{U}_\Omega$ around the infinity
 \begin{equation*}
  |\hat G(\zeta,L)| \leq \frac{K}{|\zeta|}\exp\left(M|\zeta|L\right).
 \end{equation*}
\end{thm}
\begin{proof}
 From Proposition \ref{prop:bound_gamma_infinity} we can find a bound of $g$ of $\hat\gamma$ and $\hat\gamma'$ which is increasing and bounded at infinity. Using such a bound in Proposition \ref{prop:main_bound} we obtain
 \begin{align*}
  |\hat G(\zeta,L)| & \leq \sum_{n=1}^\infty[(7g(\zeta)+3)|\zeta|]^{n-1} g(\zeta)\frac{L^n}{n!} \\
		    & = \frac{g(\zeta)}{(7g(\zeta)+3)|\zeta|}{\Big(\exp\left[(7g(\zeta)+3)|\zeta|L\right] - 1\Big)} \\
		    & \leq \frac{{K}}{|\zeta|}\exp(M|\zeta|L)
 \end{align*}
 for some $K>0$, and where we have set $M:=7{[\sup_{\zeta\in\mathcal{U}_\Omega}g(\zeta)]}+3{<\infty}$ {since we have assumed $g$ to be bounded at infinity}. 
\end{proof}
This result, together with Theorem \ref{thm:resurgence_two_points_function}, directly implies
\begin{coro} \label{coro:BE_res}
 The solution of the renormalisation group equation \eqref{eq:RGE} and the Schwinger-Dyson equation \eqref{eq:SDnlin} is Borel-\'Ecalle resummable.
\end{coro}
The main Theorem \ref{thm:main} is obtained with one more result.
\begin{prop} \label{prop:numerical_bound}
 The Borel-\'Ecalle resummed function $G^{\rm res}(a,L)$ is analytic {in} the open subset of $\C$ defined by
 \begin{equation*}
  \left|a-\frac{1}{20L}\right| < \frac{1}{20L}
 \end{equation*}
 for any $L$ in $\R^*_+$.
\end{prop}
\begin{proof}
 The analyticity domain of the resummed function only depends on the asymptotic of the Borel transform. We can therefore subtract to $\hat\gamma$ a function $\psi$ with a compact support 
 without changing the analyticity domain. Doing this, one can assume that the bound $g$ of Proposition \ref{prop:main_bound} is {bounded} at infinity by the function $G$. In this case we have
 \begin{equation*}
  \sup_{\zeta\in\mathcal{U}_\Omega}g(\zeta) \leq \sup_{X\in\R_+}f(X)=f(0)=1
 \end{equation*}
 and therefore $M\leq10$. {The result then follows from} Theorem \ref{thm:Borel_Ecalle_resummation}.
\end{proof}
Let us finish this article by pointing out that we have shown the analyticity of a solution of the Schwinger-Dyson equation in an open disc tangent to the origin. {Assuming that the bound of Theorem \ref{thm:bound_two_point_infinity} is optimal, standard results of the theory of Laplace transform and of Borel-\'Ecalle resummation theory indicate that the resummed function $G^{\rm res}(a,L)$ admits a logarithmic singularities at $a=(10L)^{-1}$. Notice that this logarithmic singularity was already pointed out in the conclusion of \cite{BeCl16}.

If one sees the resummed function $G^{\rm res}(a,L)$ as a function of $p^2=\mu^2\exp(L)$, its singularities at finite $p^2$ can be seen as masses that were not present in the lagrangian but can only be seen after a resurgent analysis. Further notice that if the Borel transform of the two-points function has an exponential behavior at infinity 
\begin{equation} \label{bound_G_asympt_free}
 \widehat{G}(\zeta,L)\sim K\exp(ML|\zeta|)
\end{equation}
then the associated resummed function admits a simple pole at $ML=1/a~\Longleftrightarrow p^2=\mu^2\exp((aM)^{-1})$. In other words: under the assumption of the bound \eqref{bound_G_asympt_free} we have generated a mass $\mu^2\exp((aM)^{-1})$ for our theory. 

Finally, let us point out two things. First, that a bound of the form \eqref{bound_G_asympt_free} is what one should expect to obtain after performing an acceleration of the Borel transform. Furthermore, according to \cite{BeCl18}
such an acceleration will likely take place in the context of asymptotically free QFTs. Therefore we are confident that the proposed mechanism could, at least in principle, be applied to some Yang-Mills theories. Second, if one improves the bound \eqref{bound_G_asympt_free}\footnote{this being of course an abuse of language: it is only possible if Equation \eqref{bound_G_asympt_free} is a bound not an equivalence. We are not more precise in order to not burden the text with too much technical details.} that is to say find an $M'<M$ then the induced mass $\mu^2\exp((aM')^{-1})$ will increase. In other words: improving the bound \eqref{bound_G_asympt_free} increases the mass gap of the theory.

This non perturbative mass generation mechanism stems from the ideas of \cite{BeCl16}, were a similar mechanism was proposed for a transseries approach of the problem. Our refined mechanism will require a finer analysis of the Laplace transform and \'Ecalle's acceleration. Such an analysis is beyond the scope of the current work but will be tackled in the future. \\
}

\noindent
{\bf Acknowledgments:} The author thanks Marc Bellon for many exciting discussions on resurgence theory and the Wess-Zumino model. I also {thank} 
David Sauzin for having kindly answered my questions regarding his non linear analysis for resurgent functions and Sylvie Paycha for encouragements,
discussions and suggestions. I am very { grateful for Marc Bellon's and Sylvie Paycha's corrections an an early} draft of this paper. {I would also like to thank the two anonymous referees whose questions and suggestions have greatly improved the quality of this paper.} {This work was partly completed while at the} Perimeter Institute.

\bibliographystyle{unsrturl}

\begin{thebibliography}{10}

\bibitem{DaZh17}
Bin~Zhang {and} Viet~Dang.
\newblock Renormalization of {Feynman} amplitudes on manifolds by spectral zeta
  regularization and blow-ups.
\newblock 12 2017.
\newblock \href {http://arxiv.org/abs/arXiv:1712.03490}
  {\path{arXiv:1712.03490}}.

\bibitem{Pa19}
Romain Pascalie.
\newblock {A Solvable Tensor Field Theory}.
\newblock 2019.
\newblock \href {http://arxiv.org/abs/1903.02907} {\path{arXiv:1903.02907}}.

\bibitem{BrKr99}
{David}~J. Broadhurst and {Dirk}~Kreimer.
\newblock Exact solutions of {D}yson--{S}chwinger equations for iterated
  one-loop integrals and propagator-coupling duality.
\newblock {\em Nucl.\ Phys.}, B 600:403--422, 2001.
\newblock \href {http://arxiv.org/abs/hep-th/0012146}
  {\path{arXiv:hep-th/0012146}}.

\bibitem{Cl14}
Pierre~J. Clavier.
\newblock Analytic results for {S}chwinger--{D}yson equations with a mass term.
{\newblock{\em Lett. Math. Phys.}, 105,}
\newblock 2015
\newblock \href {http://arxiv.org/abs/1409.3351} {\path{arXiv:1409.3351}},
  \href {http://dx.doi.org/10.1007/s11005-015-0762-1}
  {\path{doi:10.1007/s11005-015-0762-1}}.

\bibitem{BeMaVa19}
Jahmall~Bersini, Alessio~Maiezza {and} Juan Carlos~Vasquez.
\newblock Resurgence of the renormalization group equation.
{\newblock{\em Annals of Physics}, 415,}
\newblock \href {http://arxiv.org/abs/1910.14507} {\path{arXiv:1910.14507}}.

\bibitem{BeCl14}
Marc~P. Bellon and Pierre~J. Clavier.
\newblock A {S}chwinger--{D}yson {E}quation in the {B}orel plane: singularities
  of the solution.
\newblock {\em Lett.\ Math.\ Phys.}, 105, 2015.
\newblock \href {http://arxiv.org/abs/1411.7190} {\path{arXiv:1411.7190}},
  \href {http://dx.doi.org/10.1007/s11005-015-0761-2}
  {\path{doi:10.1007/s11005-015-0761-2}}.

\bibitem{BeCl16}
Marc~P. Bellon and Pierre~J. Clavier.
\newblock Alien calculus and a {S}chwinger--{D}yson equation: two-point function with a nonperturbative mass scale.
\newblock 2016.
\newblock {\em Lett.\ Math.\ Phys.}, 108 (2) pp.391-412.
\newblock \href {http://arxiv.org/abs/1612.07813 [hep-th]}
  {\path{arXiv:1612.07813 [hep-th]}}
\newblock 10.1007/s11005-017-1016-1.


\bibitem{Ecalle81}
Jean \'Ecalle.
\newblock {\em Les fonctions r\'esurgentes, Vol.1}.
\newblock Pub. Math. Orsay, 1981.

\bibitem{Ecalle81b}
Jean \'Ecalle.
\newblock {\em Les fonctions r\'esurgentes, Vol.2}.
\newblock Pub. Math. Orsay, 1981.

\bibitem{Ecalle81c}
Jean \'Ecalle.
\newblock {\em Les fonctions r\'esurgentes, Vol.3}.
\newblock Pub. Math. Orsay, 1981.

\bibitem{Ec92}
Jean \'Ecalle.
\newblock {\em Introduction aux fonctions analysables et preuve constructive de
  la conjecture de {D}ulac}.
\newblock Hermann, 1992.

\bibitem{Me96}
Frédéric Menous.
\newblock {\em Les bonnes moyennes uniformisantes et leurs applications a la
  resommation reelle}.
\newblock PhD thesis, 1996.
\newblock Thèse de doctorat dirigée par \'Ecalle, Jean Sciences et techniques
  communes Paris 11 1996.
\newblock URL: \url{http://www.theses.fr/1996PA112392}.

\bibitem{Menous}
Fr\'ed\'eric Menous.
\newblock Les bonnes moyennes uniformisantes et une application \`a la
  resommation r\'eelle.
\newblock {\em Annales de la Facult\'e des sciences de Toulouse :
  Math\'ematiques}, 6e s{\'e}rie, 8(4):579--628, 1999.
\newblock URL: \url{http://www.numdam.org/item/AFST_1999_6_8_4_579_0}.

\bibitem{VB14}
Emmanuel Vieillard-Baron.
\newblock {\em From resurgent functions to real resummation through
  combinatorial Hopf algebras}.
\newblock PhD thesis, 2014.
\newblock Thèse de doctorat dirigée par Rolin, Jean-Philippe Mathématiques
  Dijon 2014.
\newblock URL: \url{http://www.theses.fr/2014DIJOS005}.

\bibitem{AnSc13}
Inês Aniceto and Ricardo Schiappa.
\newblock Nonperturbative ambiguities and the reality of resurgent transseries.
\newblock {\em Communications in Mathematical Physics}, 335:183--245, 2013.
\newblock \href {http://arxiv.org/abs/1308.1115} {\path{arXiv:1308.1115}}.

\bibitem{AnBaSc18}
Ricardo~Schiappa, Inês~Aniceto {and} Gökçe~Başar.
\newblock A primer on resurgent transseries and their asymptotics.
\newblock {\em Physics Reports}, 809, 02 2018.
\newblock \href {http://dx.doi.org/10.1016/j.physrep.2019.02.003}
  {\path{doi:10.1016/j.physrep.2019.02.003}}.

\bibitem{Do14}
Daniele Dorigoni.
\newblock An introduction to resurgence, trans-series and alien calculus.
\newblock {\em Annals of Physics}, 11 2014.
\newblock \href {http://dx.doi.org/10.1016/j.aop.2019.167914}
  {\path{doi:10.1016/j.aop.2019.167914}}.

\bibitem{Sa12}
David Sauzin.
\newblock Nonlinear analysis with resurgent functions.
\newblock 2012.
\newblock \href {http://arxiv.org/abs/1212.4477v4} {\path{arXiv:1212.4477v4}}.


\bibitem{BeCl18} {
Marc P. Bellon and Pierre J. Clavier.
\newblock{Analyticity domain of a Quantum Field Theory and Accelero-summation},
\newblock{\em{Lett. in Math. Phys.}, Volume 109},
\newblock{2019},
\newblock{DOI: 10.1007/s11005-019-01172-0},
\newblock{\href {https://arxiv.org/abs/1806.08254}
{\path{arXiv:1806.08254}}}.
}

\bibitem{So80}{
Alan D. Sokal. \newblock{An improvement of Watson’s theorem on Borel summability}, \newblock{{\em Journal of Mathematical Physics},
21 (2), pp 261-263 (1980).}
doi = 10.1063/1.524408.}

\bibitem{Co98} {
Ovidiu Costin.
\newblock{On Borel summation and Stokes phenomena for rank-$1$ nonlinear systems of ordinary differential equations},
\newblock{{\em Duke Mathematical Journal}, 93 (2),}
\newblock{1998},
\newblock{DOI: 10.1215/S0012-7094-98-09311-5},
\newblock{\href {https://arxiv.org/abs/math/0608408}
{\path{arXiv:math/0608408}}}.
}

\bibitem{Co06} 
{Ovidiu Costin.
\newblock{Exponential asymptotics, trans-series and generalized Borel summation for analytic nonlinear rank one systems of ODE’s},
\newblock{\href {https://arxiv.org/abs/math/0608414}
{\path{arXiv:math/0608414}}}}
\bibitem{Me97}
Fr\'ed\'eric Menous.
\newblock The well-behaved catalan and brownian averages and their applications
  to real resummation.
\newblock {\em Proceedings of the Symposium on Planar Vector Fields (Lleida,
  1996). Publ. Mat.}, 41:209—222, 1997.

\bibitem{Bo11}
Olivier Bouillot.
\newblock {\em {Invariants Analytiques des Diff\'eomorphismes et
  MultiZ\^etas}}.
\newblock PhD thesis, Universit\'e Paris-Sud 11, 2011.
\newblock URL: \url{http://tel.archives-ouvertes.fr/tel-00647909}.

\bibitem{Sa14}
David Sauzin.
\newblock Introduction to 1-summability and resurgence.
\newblock 2014.
\newblock \href {http://arxiv.org/abs/1405.0356v1} {\path{arXiv:1405.0356v1}}.

\bibitem{KaSa16}
David~Sauzin {and} Shingo~Kamimoto.
\newblock Iterated convolutions and endless {Riemann} surfaces.
{\newblock{\em Annali Scuola Normale Superiore - Classe di Scienze}, 20 (1),}
\newblock{2016},
{\newblock{DOI: 10.2422/2036-2145.201708 008},}
\newblock \href {http://arxiv.org/abs/1610.05453v2}
  {\path{arXiv:1610.05453v2}}.

\bibitem{VB19}
Emmanuel Viellard-Baron.
\newblock \'Ecalle's averages, {Rota-Baxter} algebras and the construction of
  moulds.
\newblock 2019.
\newblock \href {http://arxiv.org/abs/arXiv:1904.02417v1}
  {\path{arXiv:1904.02417v1}}.
  
  
\bibitem{BoDu20} {Michael Borinsky and Gerald V. Dunne.
\newblock{Non-perturbative completion of Hopf-algebraic Dyson-Schwinger equations},
\newblock{\em Nuclear Physics B}, Volume 957,
\newblock{2020},
\newblock{115096,
ISSN 0550-3213,
https://doi.org/10.1016/j.nuclphysb.2020.115096.}}

\bibitem{EcMe95} {
Jean \'Ecalle and Fr\'ed\'eric Menous.
\newblock{Well-behaved convolution averages and the non-accumulation theorem for limit-cycles},
\newblock{in {\em The Stokes Phenomenon and Hilbert's 16th Problem}}, 
\newblock{https://doi.org/10.1142/3031}.
}

\bibitem{Cl15}
Pierre~J. Clavier.
\newblock {\em Analytic and Geometrical approches of non-perturbative quantum
  field theories}.
\newblock PhD thesis, 2015.

\bibitem{WeZu74a}
Julius Wess and Bruno Zumino.
\newblock Supergauge transformations in four dimensions.
\newblock {\em Nucl.\ Phys.\ B}, 70:39--50, 1974.

\bibitem{WeZu74b}
Bruno~Zumino {and} Julius~Wess.
\newblock A lagrangian model invariant under supergauge transformations.
\newblock {\em Phys.\ Lett.}, 49B:52--55, 1974.

{
\bibitem{Co08}
Ovidiu Costin. 
\newblock{\em Asymptotics and Borel summability}, 
\newblock{Monographs and Surveys in Pure and Applied Math (2008)},
\newblock{Chapman and Hall/CRC},
\newblock{ISBN 13:
9781420070316.}}


\bibitem{CoTa07} {
Ovidiu Costin and Saleh Tanveer, \newblock{Nonlinear evolution PDEs in $\R^+\times\C^d$ existence and uniqueness of solutions, asymptotic and Borel summability properties}, \newblock{\em Ann. I. H. Poincaré} AN 24 (2007).
}

\bibitem{Be10a}
{M}arc~P. {B}ellon.
\newblock An efficient method for the solution of {S}chwinger--{D}yson
  equations for propagators.
\newblock {\em Lett.\ Math.\ Phys.}, 94:77--86, 2010.
\newblock \href {http://arxiv.org/abs/1005.0196} {\path{arXiv:1005.0196}},
  \href {http://dx.doi.org/10.1007/s11005-010-0415-3}
  {\path{doi:10.1007/s11005-010-0415-3}}.

\bibitem{BeCl13}
Marc~P. Bellon and Pierre~J. Clavier.
\newblock Higher order corrections to the asymptotic perturbative solution of a
  {S}chwinger--{D}yson equation.
\newblock {\em Lett.\ Math.\ Phys.}, 104:1--22, 2014.
\newblock \href {http://arxiv.org/abs/1311.1160v2} {\path{arXiv:1311.1160v2}},
  \href {http://dx.doi.org/10.1007/s11005-014-0686-1}
  {\path{doi:10.1007/s11005-014-0686-1}}.

\bibitem{BeCl15}
Marc P.~Bellon {and} Pierre J.~Clavier.
\newblock Solving a {Dyson–Schwinger} equation around its first singularity in
  the {Borel} plane.
 {\newblock{\em Front.Phys.}, 11 (6),
 \newblock{2016},
 \newblock{10.1007/s11467-016-0582-5}.}

\bibitem{BeLoSc07}
{M}arc {B}ellon, {G}ustavo {L}ozano and {F}idel {S}chaposnik.
\newblock {H}igher loop renormalization of a supersymmetric field theory.
\newblock {\em {P}hysics {L}etters {B}}, 650:293--297, 2007.
\newblock  \href
  {http://arxiv.org/abs/arXiv:hep-th/0703185}
  {\path{arXiv:arXiv:hep-th/0703185}}, \href
  {http://dx.doi.org/10.1016/j.physletb.2007.05.024}
  {\path{doi:10.1016/j.physletb.2007.05.024}}.

\bibitem{Ho79}
{Gerard}~'t~Hooft.
\newblock {\em Can We Make Sense Out of ``Quantum Chromodynamics''?}, pages
  943--982.
\newblock Springer US, Boston, MA, 1979.
\newblock URL: \url{http://dx.doi.org/10.1007/978-1-4684-0991-8_17}, \href
  {http://dx.doi.org/10.1007/978-1-4684-0991-8_17}
  {\path{doi:10.1007/978-1-4684-0991-8_17}}.

\bibitem{Hormander66}
Lars Hörmander.
\newblock {\em An introduction to complex analysis in several complex
  variables}.
\newblock Elsevier, 1966.


\bibitem{Be10} {Marc P. Bellon.
\newblock{Approximate differential equations for renormalization group functions in models free of vertex divergencies},
\newblock{\em Nuclear Physics B}, Volume 826, Issue 3,
\newblock{2010},
\newblock{Pages 522-531,
ISSN 0550-3213,
https://doi.org/10.1016/j.nuclphysb.2009.11.002.}}



\end{thebibliography}

\end{document}